\documentclass[a4paper, intlimits, 11pt,reqno]{amsart}
\usepackage[T1]{fontenc}
\usepackage[english]{babel}
\usepackage{graphicx}
\usepackage{float}
\usepackage{bm}
\usepackage{subcaption}
\usepackage{booktabs}
\usepackage{threeparttable}  
\usepackage{multirow}
\usepackage{eurosym}
\usepackage[sort&compress, numbers]{natbib}

\usepackage{amsthm, amssymb, amsfonts}
\usepackage{amsmath}
\usepackage{enumitem}
\usepackage{color}
\usepackage[table, dvipsnames]{xcolor}
\usepackage{accents}
\usepackage{tikz}
\usepackage{tkz-graph}
\usepackage{makecell}

\usetikzlibrary{calc,positioning, shapes, fit,arrows,arrows.meta,decorations.pathreplacing,calligraphy}
\usepackage{comment}
\usepackage{url}

\DeclareMathOperator*{\argmax}{arg\,max}
\usepackage{hyperref}
\hypersetup{
 colorlinks  = true, %
 urlcolor   = black, %
 linkcolor  = black, %
 citecolor  = RoyalBlue %
}

\newtheorem{theorem}{Theorem}[section]
\newtheorem{corollary}[theorem]{Corollary}

\newtheorem{proposition}[theorem]{Proposition}

\theoremstyle{definition}
\newtheorem{definition}{Definition}[section]
\newtheorem{assumption}[definition]{Assumption}

\theoremstyle{remark}
\newtheorem{remark}{Remark}[section] 
\newtheorem{example}[remark]{Example}

\numberwithin{equation}{section}

\begin{document}

\title{Comparison of Tax and Cap-and-Trade Carbon Pricing Schemes}

\author{St\'ephane Cr\'epey}
\address{Laboratoire de Probabilités, Statistique et Modélisation (LPSM), Sorbonne Université et Université Paris Cit\'e, CNRS UMR 8001.}
\email{stephane.crepey@lpsm.paris}
\urladdr{https://perso.lpsm.paris/crepey}

\author{Samuel Drapeau}
\address{School of Mathematical Sciences and Shanghai Advanced Institute for Finance, Shanghai Jiao Tong University, Shanghai, China}
\email{sdrapeau@saif.sjtu.edu.cn}
\urladdr{http://www.samuel-drapeau.info}

\author{Mekonnen Tadese}
\address{Laboratoire de Probabilités, Statistique et Modélisation (LPSM), Sorbonne Université and 
Mathematics Department, Woldia University, Woldia, Ethiopia}
\email{demeke@lpsm.paris}

\date{\today}

\thanks{
	This research has benefited from the support of Chair Capital Markets Tomorrow: Modeling and Computational Issues under the aegis of the Institute Europlace de Finance,  a joint initiative of Laboratoire de Probabilit\'es, Statistique et Modélisation (LPSM) / Université Paris Cit\'e and  Cr\'edit Agricole CIB,  and of Chair Stress Test, Risk Management and Financial Steering, led by the French Ecole Polytechnique and its foundation and sponsored by BNP Paribas}

\begin{abstract}
	Carbon pricing has become a central pillar of modern climate policy, with carbon taxes and emissions trading systems (ETS) serving as the two dominant approaches.
	Although economic theory suggests these instruments are equivalent under idealized assumptions, their performance diverges in practice due to real-world market imperfections.
	A particularly less explored dimension of this divergence concerns the role of financial intermediaries in emissions trading markets.
	This paper develops a unified framework to compare the economic and environmental performance of tax- and market-based schemes, explicitly incorporating the involvement of financial intermediaries.
	By calibrating both instruments to deliver identical aggregate emission reduction targets, we assess their economic performance across alternative market structures.
	Our results suggest that although the two schemes are equivalent under perfect competition, the presence of intermediaries in ETS reduces both regulatory wealth and the aggregate wealth of economic agents relative to carbon taxation. 
	These effects stem from intermediaries' influence on price formation and their appropriation of part of the revenue stream.
	The findings underscore the importance of accounting for intermediaries' behavior in the design of carbon markets and highlight the need for further empirical research on the evolving institutional structure of emissions trading systems.

	\vspace{5pt}
	\vspace{5pt}

\end{abstract}

\maketitle

\section{Introduction}
Carbon pricing has emerged as a cornerstone of global climate policy, primarily implemented through taxes and emissions trading systems (ETS).
These mechanisms aim to internalize the social cost of greenhouse gas emissions by assigning them a monetary value \citep{Worldbank2025}.
Carbon taxes establish a fixed price per unit of emissions, offering price certainty and enabling companies to determine their own abatement levels in response to the price signal. 
In contrast, ETS---often referred to as cap-and-trade or market scheme---sets a fixed emissions cap and allows the market to determine the price of emission certificates through trading, thereby ensuring environmental certainty \citep{carmona2010market,hitzemann2018,aid2023,Worldbank2025}.
As of 2025, a total of 80 jurisdictions worldwide have implemented these mechanisms, including 43 carbon taxes (covering 5\% of global emissions) and 37 ETS (covering 23 \%), together accounting for roughly 28\% of global greenhouse gas emissions \citep{Worldbank2025, ICAP2025}.

Although carbon taxes and emissions trading systems have been widely adopted, their comparative effectiveness remains contested.
A foundational result by \citep{Weitzman1974} demonstrates that under idealized conditions of complete information and no uncertainty, price-based instruments (taxes) and quantity-based instruments (market schemes) are equivalent.
Similarly, \citep{anand2020} shows that these two instruments remain equivalent under oligopolistic competition when emissions are certain.
In practice, market frictions such as uncertainty, transaction costs, and incomplete information weaken this equivalence \citep{Weitzman1974,islegen2016}.
Empirical and theoretical evidence is mixed: several studies report superior performance of carbon taxes in specific contexts \citep{MacKay2015, Ahmad2024, Fan2023}.
By contrast, market schemes enable companies with low marginal abatement costs to reduce emissions and sell surplus allowances to companies facing higher costs, making them widely promoted as cost-effective instruments \citep{cludius2022, montgomery1972, Xu2023}.
These perceived advantages often receive prominent media coverage \citep{Kosnik2018}.
Moreover, political considerations have increasingly favored emissions market over carbon tax  in recent years \citep{IMF2022}.

A distinctive characteristic of market schemes is the active engagement of diverse financial actors alongside regulated compliance entities.
The majority of these participants are financial intermediaries, such as banks, credit institutions, brokers, and investment companies \citep{Lucia2015, Roques2022, cludius2022, Quemin2023, isah2024,EEX2025}, who purchase emission certificates at auctions or in secondary markets and subsequently resell them to compliance clients.
For example, in 2023, non-compliant financial actors dominated participation in the EU ETS, acquiring approximately 80\% of auctioned allowances and accounting for 56\% of secondary-market trading volume \citep{ESMA2024}.
Recent research underscores that the activities of such financial actors significantly shape emission certificates price formation \citep{heindl2012, Lucia2015, Johanna2020, cludius2022, Quemin2023, isah2024}, whereas earlier studies primarily emphasized on the fundamentals such as the emissions cap, penalty rates, and green technology costs \citep{carmona2009, carmona2010market, hitzemann2018, aid2023}.

The role of financial intermediaries in carbon markets remains a topic of ongoing debate.
On the one hand, financial intermediaries act as market makers, enabling smaller or less experienced compliance entities to access the market \citep{ESMA2024, cludius2022, Johanna2020, Lucia2015}.
On the other hand, critics argue that intermediaries may amplify price volatility, and their dominance in auctions raises concerns about potential market manipulation \citep{heindl2012, Quemin2023, isah2024, cludius2022}.
Although the expanding role of intermediaries is well documented, comparative theoretical frameworks quantifying their net economic effects remain underdeveloped.
In particular, existing analyses provide limited insight into how wealth-maximizing intermediary behavior redistributes wealth among companies, regulator, and intermediaries, or how these outcomes diverge from those under a tax scheme.

This paper seeks to compare tax and market schemes by explicitly accounting for the role of financial intermediaries in the latter.
Specifically, it examines how these two approaches differ in the distribution of emissions among compliance companies and in their broader economic impacts (e.g.\ GDP contributions) on regulated companies, financial intermediaries, and regulator.

We develop a one-period equilibrium model to compare tax and market schemes, with particular emphasis on the role of financial intermediaries in the latter.
Compliance companies produce goods and emit carbon, modeled first as deterministic and later as uncertain, with all emissions priced (no free certificates).
Under the tax scheme, companies reduce emissions  via green technologies or pay a tax for each unit of carbon they emit.
Under the market scheme, certificates are auctioned; companies purchase either directly or through intermediaries at company-specific prices, facing penalties for excess emissions. To isolate the impacts of intermediaries under the market scheme, we analyze two extremes: all companies buy at auction or all use intermediaries.
To make the schemes comparable, we calibrate the tax and penalty to produce identical aggregate emissions.

Our main findings are as follows.
First, under perfect competition---where all companies purchase emission certificates directly at auctions organized by the regulatory authority in the market scheme---tax and market mechanisms produce identical outcomes.
This equivalence holds irrespective of whether emission is modeled as deterministic or stochastic.
In particular, abatement levels and compliance costs coincide across the two schemes once the tax rate and penalty are calibrated to match the emissions cap.
Thus, in the absence of intermediaries, the distinction between a tax and a market scheme becomes immaterial.

Second, the introduction of financial intermediaries into the emissions market alters the distribution of economic value, reducing the GDP relative to the tax regime.
When companies rely entirely on these intermediaries to access the allowance market---instead of participating directly---the combined economic output of companies, regulator, and intermediaries under the market scheme is strictly lower than  under the tax scheme.
Formally, when emission is deterministic, we establish that
\begin{equation*}
	\underbrace{ \mathcal{W}^{\mathrm{tax}}_{\mathrm{C}} + \mathcal{W}^{\mathrm{tax}}_{\mathrm{R}} }_{\mathrm{GDP^{\mathrm{tax} }}}  
	   			 > 
	\underbrace{  \mathcal{W}^{\mathrm{mar}}_{\mathrm{C}} + \mathcal{W}^{\mathrm{mar}}_{\mathrm{R}} + \mathcal{W}_{\mathrm{F}}^{\mathrm{mar}} }_{\mathrm{GDP^{\mathrm{mar} }}},
\end{equation*}
where $\mathcal{W}^{\bm{\mathrm{\cdot}}}_{\mathrm{C}}$, $\mathcal{W}^{\bm{\mathrm{\cdot}}}_{\mathrm{R}}$ and $\mathcal{W}_{\mathrm{F}}^{\bm{\mathrm{\cdot}}}$ denote the aggregate wealth of the companies, of the regulatory authority, and of the financial intermediaries under the market ($\mathrm{mar}$) or tax scheme ($\mathrm{tax}$).
Moreover, the aggregate wealth of companies is lower under the market scheme:
\begin{equation*}
	\mathcal{W}^{\mathrm{tax}}_{\mathrm{C}} \geq \mathcal{W}^{\mathrm{mar}}_{\mathrm{C}}.
\end{equation*}
This decline is driven by the market power of intermediaries, which allows them to extract rents.
However, the wealth of an individual company under the market scheme may be higher or lower than its wealth under the tax scheme, depending on specific factors such as its green investment costs.
In the case of random emissions, comparable results are seen to hold numerically.

Third, in the presence of intermediaries, the wealth of the regulatory authority under the market scheme is strictly lower than under the tax scheme.
For deterministic emissions, we show
\begin{equation*}
	\mathcal{W}^{\mathrm{tax}}_{\mathrm{R}} 
			>
	\mathcal{W}^{\mathrm{mar}}_{\mathrm{R}} = \mathcal{W}^{\mathrm{tax}}_{\mathrm{R}} - \frac{A^2}{\varrho},
\end{equation*}
where $A$ denotes the number of certificates auctioned and $\varrho$ is an aggregate factor depending on the companies environmental costs.
Again, in the case of random emissions, comparable results are seen to hold numerically.
Intermediaries depress spot market prices and reallocate a portion of auction wealth to financial actors, thereby altering the overall distribution of wealth.
The above finding aligns with recent evidence \citep[pp. 40--41]{Worldbank2025}, which reports a 10\% decline in regulator wealth under market schemes in 2024 compared with 2023, whereas wealth under tax schemes increased by 10\% over the same period.
The report attributes this decline primarily to lower certificate prices.
Our analysis suggests that increased intermediary participation may be one contributing factor to this divergence, though further empirical investigation is needed to confirm this mechanism.

The remainder of the paper is structured as follows.
Section~\ref{sec:model} introduces the models for tax and market schemes.
Section~\ref{sec:compare} presents a theoretical comparison of economic outcomes under the tax and market schemes.
Section~\ref{sec:numeric} provides numerical illustrations.
Section~\ref{sec:princing_random} extends the analysis to settings with emissions uncertainty. 
Section~\ref{sec:conclusion} concludes.
Proofs and supplementary material are provided in Sections~\ref{app:model}--\ref{app:random_case}.

Equilibrium quantities are denoted by \textbf{bold letters} throughout the paper, e.g.\ $q_i$ denotes a quantity of goods produced by a company $i$ whereas 
 $\bm{q}_i$ denotes its optimal production.
 
\section{Model} \label{sec:model}
The model for emissions regulation involves different agents:
\begin{itemize}
	\item \textbf{Companies}: They produce goods for wealth, emitting carbon as a by-product.
	\item \textbf{Regulatory authority}: It is responsible for setting the overall emissions standard as well as the policy scheme.
	      They either collect a tax on carbon emission or issue carbon certificates and levy a penalty on excess emissions.
	\item \textbf{Financial intermediaries}: Under a market scheme, they sell certificates to their clients while acquiring them from the regulator on an auction basis.
\end{itemize}

We compare two distinct regulatory schemes for managing carbon emissions:
\begin{itemize}
	\item \textbf{Tax}: Companies   are charged a tax for each unit of carbon they emit, creating a direct cost for pollution and an incentive to reduce emissions.
	\item \textbf{Market}: This scheme involves the establishment of a market for emission certificates.
	      Companies can hold a certain number of certificates to cover their emissions.
	      If a company's emissions exceed the number of certificates it owns, it pays a penalty for each excess unit of carbon emitted.
	      Companies   can purchase these certificates either directly on the auction organized by the regulator or through intermediaries.
\end{itemize}

\subsection{Companies}
A generic company produces $q$ units of a good for a raw wealth (ignoring any carbon-related factors)
\begin{equation*}
	\pi(q) = \pi^0 q -\frac{\pi^1}{2}q^2,
\end{equation*}
where $\pi^0$ and $\pi^1$ are positive company-specific constants.
In view of this work, we assume that the num\'eraire of a production unit is given in terms of the amount of carbon it generates.
Hence, for a given production level $q$, the total carbon emissions are exactly $q$.
By investing in green technologies, the company can reduce its total emissions $q$ to $e^{-a}q$, where $a > 0$ is the company's green technology effort.
The cost of this green investment as a function of the emission reductions volume $(1 - e^{-a} ) q$ is given by
\begin{equation*}
	c(q, a)  = \frac{\gamma}{2} \left[ (1-e^{-a} ) q \right]^2,
\end{equation*}
where $\gamma>0$ is a company-specific green cost factor.

Under the tax scheme, a tax $\tau$ is levied per unit of emissions, resulting in a net wealth
\begin{equation}\label{eq:tax_wealth}
	\pi(q) - c(q,a) - \tau e^{-a} q.
\end{equation}
Under the market scheme, the company can purchase an amount $\delta$ of emission certificates at a unit price $P$.
For each unit of emitted carbon exceeding the amount of certificates $\delta$ owned, a penalty $\lambda$ is levied, resulting in a net wealth
\begin{equation}\label{eq:market_wealth}
	\pi(q) -  c(q, a) - \delta P - \lambda \left(q e^{-a}- \delta \right)^+.
\end{equation}

Under each regulatory policy, the company's objective is to maximize its overall wealth.
As a baseline, in the business-as-usual case, i.e.  in the absence of any carbon emission constraints,  the wealth-maximizing level of production $\bm{ q^{\mathrm{bau}}  }$ and the corresponding carbon emissions $\mathrm{E}^{\mathrm{bau}}$ and wealth $\mathcal{W}^{\mathrm{bau}}_{\mathrm{C}}$ of the company are given by
\begin{align}
	\label{eq:bau}
	\bm{ q^{\mathrm{bau}}  }             & = \frac{\pi^0}{\pi^1},
	                                     &
	\mathrm{E}^{\mathrm{bau}}            & = \bm{ q^{\mathrm{bau}}  },
	                                     &
	\mathcal{W}^{\mathrm{bau}}_{\mathrm{C}} & = \pi(\bm{ q^{\mathrm{bau}}  }) = \frac{1}{2}\frac{(\pi^0)^2}{\pi^1}.
\end{align}
Furthermore, for ease of notation, we define the emission factor
\begin{equation*}
	\varrho = \frac{1}{\pi^1} + \frac{1}{\gamma}.
\end{equation*}

\begin{remark}
	For feasibility reasons, in the market case, we always assume that $0<P< \lambda$.
	Indeed, for a price $P\geq \lambda$,  it is always more profitable for a company to set $\delta = 0$ and pay the penalty $\lambda$.
\end{remark}

\begin{assumption}\label{ass:standing_ass}
	A company seeks to maximize its wealth subject to the constraint that both production and wealth are positive.
	To simplify the analysis and avoid resorting to more advanced mathematical techniques, such as the use of Lagrangian multipliers, we impose the following sufficient condition:
	\begin{equation*}
		\tau \varrho < \mathrm{E}^{\mathrm{bau}} \quad \text{and} \quad \lambda \varrho < \mathrm{E}^{\mathrm{bau}}.
	\end{equation*}
\end{assumption}

\begin{proposition}\label{prop:optimalquant}
	Let $0 < P < \lambda$ be given, representing the unit certificate price faced by the company.
	Under Assumption~\ref{ass:standing_ass}, for the tax scheme $(\mathrm{tax})$ and the market scheme $(\mathrm{mar})$ it holds
	\begin{itemize}[fullwidth]
		\item $\bm{\mathrm{tax}}$: the optimal production $\bm{ q^{\mathrm{tax}} }$ and reduction rate $\bm{ a^{\mathrm{tax}} }$ are given by
		      \begin{equation*}
			      \bm{ q^{\mathrm{tax}} }      = \bm{ q^{\mathrm{bau}}  }- \tau\frac{1}{\pi^1},
			      \qquad
			      e^{-\bm{ a^{\mathrm{tax}} }}  = \displaystyle \frac{\mathrm{E}^{\mathrm{bau}} - \tau \varrho}{ \bm{ q^{\mathrm{tax}} }}
		      \end{equation*}
		      for the resulting carbon emissions $\mathrm{E}^{\mathrm{tax}}$ and wealth $\mathcal{W}^{\mathrm{tax}}_{\mathrm{C}}$
		      \begin{equation*}
			      \mathrm{E}^{\mathrm{tax}}      = \mathrm{E}^{\mathrm{bau}}- \tau \varrho,
			      \qquad
			      \mathcal{W}^{\mathrm{tax}}_{\mathrm{C}}  = \mathcal{W}^{\mathrm{bau}}_{\mathrm{C}} - \tau \left(\mathrm{E}^{\mathrm{bau}} - \frac{\tau}{2} \varrho \right)
		      \end{equation*}

		\item $\bm{\mathrm{mar}}$: the optimal production $\bm{ q^{\mathrm{mar}} }$, reduction rate $\bm{ a^{\mathrm{mar} } }$ and certificates demand $\delta$ are given by
		      \begin{equation*}
			      \bm{ q^{\mathrm{mar}} }    = \bm{ q^{\mathrm{bau}}  }- P \dfrac{1}{\pi^1},
			      \qquad
			      e^{-\bm{ a^{\mathrm{mar} } }}  = \displaystyle \frac{ \mathrm{E}^{\mathrm{bau}} - P \varrho}{\bm{ q^{\mathrm{mar}} }},
			      \qquad
			      \bm{\delta}    = \mathrm{E}^{\mathrm{bau}} -  P \varrho
		      \end{equation*}
		      for the resulting carbon emissions $\mathrm{E}^{\mathrm{mar}}$ and wealth $\mathcal{W}^{\mathrm{mar}}_{\mathrm{C}}$
		      \begin{equation*}
			      \mathrm{E}^{\mathrm{mar}}   = \mathrm{E}^{\mathrm{bau}} - P \varrho =\bm{\delta},
			      \qquad
			      \mathcal{W}^{\mathrm{mar}}_{\mathrm{C}}  = \mathcal{W}^{\mathrm{bau}}_{\mathrm{C}} - P \left(\mathrm{E}^{\mathrm{bau}} - \frac{P}{2} \varrho \right).
		      \end{equation*}
	\end{itemize}
\end{proposition}

\begin{proof}
	Under the tax scheme, the company's wealth function~\eqref{eq:tax_wealth} is given by
	\begin{equation*}
		f(q,a) = \pi^0 q - \frac{\pi^1}{2} q^2 - \frac{\gamma}{2} \big((1-e^{-a})q\big)^2 - \tau e^{-a} q.
	\end{equation*}
	The corresponding first-order condition is
	\begin{align*}
		f_q & = \pi^0 - \pi^1 q - \gamma q (1 - e^{-a})^2 - \tau e^{-a} = 0, \\[1ex]
		f_a & = q e^{-a} \left(\tau - \gamma q (1-e^{-a}) \right) = 0.
	\end{align*}
	Solving this system yields two critical points $(q_1,a_1)$ and $(q_2,a_2)$ with
	\begin{equation*}
		q_1 = 0, \quad e^{-a_1} = \frac{\pi^0}{\tau}, \quad  and \quad q_2 = \frac{\tau}{\gamma(1-e^{-a_2})}, \quad  e^{-a_2} = \frac{\bm{ q^{\mathrm{bau}}  }- \tau \varrho}{\bm{ q^{\mathrm{bau}}  }- \tfrac{\tau}{\pi^1}},
	\end{equation*}
	where $\bm{q^{\mathrm{bau}} }$ is defined in~\eqref{eq:bau}.
	The Hessian matrix of $f$ is
	\begin{equation*}
		\nabla^2 f =
		\begin{pmatrix}
			f_{qq} & f_{qa} \\
			f_{aq} & f_{aa}
		\end{pmatrix},
	\end{equation*}
	with
	\begin{align*}
		f_{qq} & = -\pi^1 - \gamma (1-e^{-a})^2,                                \\
		f_{qa} & = f_{aq} = -2 \gamma q (1-e^{-a}) e^{-a} + \tau e^{-a},        \\
		f_{aa} & = -q e^{-a} (\tau - \gamma q (1-e^{-a})) - \gamma q^2 e^{-2a}.
	\end{align*}
	At $(q_1, a_1)$, the determinant of the Hessian satisfies $|\nabla^2 f(q_1, a_1)| = -(\pi^0)^2 < 0$, indicating that $(q_1,a_1)$ is a saddle point.
	For the second critical point $(q_2, a_2)$, substituting the identity $\tau = \gamma q_2 (1-e^{-a_2})$ into the second-order derivatives yields
	\begin{align*}
		f_{qq}(q_2, a_2)       & = -\pi^1 \left( 1 + \frac{\tau^2}{\gamma} (\pi^0 - \tau) \right) < 0, \\
		|\nabla^2 f(q_2, a_2)| & = (e^{-a_2})^2 \gamma \pi^1 q_2^2 > 0,
	\end{align*}
	where the first inequality holds because $\pi^0 > \tau$, a relation implied by Assumption~\ref{ass:standing_ass}.
	Therefore, $(q_2,a_2)$ is the unique local maximizer of $f$, see \citep[pp. 924]{stewart2008}.

	Since $f$ is coercive\footnote{
		A function $\mathbb{R} \ni x \stackrel{\ell}{\mapsto} \ell(x) \in \mathbb{R}$ is said to be coercive iff $ \displaystyle \lim_{|x|\nearrow \infty} \ell(x) =-\infty$.
	} in $q$ no maximizer can occur at $|q|\to\infty$.
	As $a\to -\infty$,  $f\to -\infty$ for any fixed $q\ne0$.
	As $a\to +\infty$, it holds
	\begin{equation*}
		f(q,a) \to \pi^0 q - \frac{1}{2}( \pi^1 + \gamma) q^2
	\end{equation*}
	whose maximum over $q$ equals $(\pi^0)^2/[2(\pi^1+\gamma)]$.
	To show $(q_2, a_2)$ is the unique global maximizer of $f$, it is enough to show 
	\begin{equation*}
		\Delta :=  f(q_2,a_2) - \frac{(\pi^0)^2}{2 (\pi^1 + \gamma)} >0. 
	\end{equation*}
	It follows that 
	\begin{align*}
		\Delta & = \frac{\gamma(\pi^0 -\tau )^2 + \pi^1\tau^2 }{2\gamma \pi^1} -  \frac{(\pi^0)^2}{2 (\pi^1 + \gamma)}  =  \frac{\left[\tau ( \gamma +\pi^1)-\gamma \pi^0 \right]^2}{2\gamma \pi^1(\gamma + \pi^1)} > 0.
	\end{align*}
	Indeed, $\Delta=0$ for $\tau = \gamma \pi^0/(\gamma+\pi^1)$. 
	But, this can't happen due to Assumption~\ref{ass:standing_ass} on $\tau$. 
	Hence, $(q_2, a_2)$ is the unique global maximizer of $f$.
	By Assumption~\ref{ass:standing_ass}, it is obvious that $a_2>0$ and $q_2>0$, which implies the optimal productions  $\bm{ q^{\mathrm{tax}} }=q_2$ and emissions reduction factor $\bm{ a^{\mathrm{tax}} } = a_2$, and  expressed as:
	\begin{equation*}
		\bm{ q^{\mathrm{tax}} }     = \bm{ q^{\mathrm{bau}}  }- \tau \dfrac{1}{\pi^1},
		\qquad
		e^{-\bm{ a^{\mathrm{tax}} }}  = \frac{ \mathrm{E}^{\mathrm{bau}} - \tau \varrho}{ \bm{ q^{\mathrm{tax}} }}.
	\end{equation*}
	Substituting the value of $\bm{ q^{\mathrm{tax}} } $ and  $\bm{ a^{\mathrm{tax}} }$  into the emission expression $\mathrm{E}^{\mathrm{tax}} = e^{-\bm{ a^{\mathrm{tax}} }}  \bm{ q^{\mathrm{tax}} }$ yields
	\begin{equation*}
		\mathrm{E}^{\mathrm{tax}}  = \mathrm{E}^{\mathrm{bau}} - \tau \varrho.
	\end{equation*}
	The resulting optimal wealth $\mathcal{W}^{\mathrm{tax}}_\mathrm{C}$ of the company is
	\begin{equation*}
		\mathcal{W}^{\mathrm{tax}}_\mathrm{C}  = f\left( \bm{ q^{\mathrm{tax}} }, \bm{ a^{\mathrm{tax}} } \right)  = \mathcal{W}^{\mathrm{bau}}_{\mathrm{C}} - \tau \left(\mathrm{E}^{\mathrm{bau}} - \frac{\tau}{2} \varrho \right).
	\end{equation*}

	Under the market scheme, and assuming that the certificate price satisfies $0 < P < \lambda$,  in view of~\eqref{eq:market_wealth}, for fixed values of $q$ and $a$, the company minimizes the cost associated with its certificate demand $\delta$:
	\begin{equation*}
		\min_{\delta} \left\{ \delta P + \lambda \left( e^{-a}  q  - \delta \right)^+ \right\} = P e^{-a}  q,
	\end{equation*}
	which is uniquely attained at $\delta =  e^{-a} q$.
	The corresponding company's wealth in~\eqref{eq:market_wealth} simplifies to
	\begin{equation*}
		\pi(q) - c(q, a) - P e^{-a} q,
	\end{equation*}
	which is equivalent to the wealth under the tax scheme, as given in~\eqref{eq:tax_wealth}, when the tax rate is set to $\tau = P$.
	By analogy, we derive the expressions for the company's optimal production level $\bm{ q^{\mathrm{mar}} }$, emissions reduction rate $\bm{ a^{\mathrm{mar} } }$, carbon emissions $\mathrm{E}^{\mathrm{mar}}$, and wealth $\mathcal{W}^{\mathrm{mar}}_{\mathrm{C}}$.
	Finally, the certificate demand simplifies to $\bm{\delta} = e^{-\bm{ a^{\mathrm{mar} } }} \bm{ q^{\mathrm{mar}} } = \mathrm{E}^{\mathrm{bau}} - P \varrho$.
\end{proof}

\begin{remark}\label{rem:zero_loss}
	As shown in Proposition~\ref{prop:optimalquant}, in the present deterministic case, the demand of certificates for any company in the market scheme exactly matches the total level of its emissions, and there is no excess emissions.
	This is no longer the case if emissions are random: see Section~\ref{sec:princing_random}.
\end{remark}

If $N$ companies, indexed by $i$ are present in the economy (see Table~\ref{tab:notations} for details on the notations), the results from the above proposition imply the following.

\begin{itemize}[leftmargin= 2em]
	\item Indexing by $i$ the results from the above proposition and introducing the notations $\mathrm{E}^{\mathrm{bau}} = \sum \mathrm{E}^{\mathrm{bau}}_i$ for aggregate carbon emissions, and $\mathcal{W}^{\mathrm{bau}}_{\mathrm{C}} = \sum \mathcal{W}^{\mathrm{bau}}_{\mathrm{C}, i}$ for aggregate wealth in the business-as-usual case, the total emissions and aggregate wealth under the tax and market schemes are
	      \begin{align*}
		      \mathrm{E}^{\mathrm{\mathrm{tax}}}
		       & = \mathrm{E}^{\mathrm{bau}} - \tau \sum \varrho_i,
		       &
		      \mathcal{W}^{\mathrm{tax}}_{\mathrm{C}}
		       & = \mathcal{W}^{\mathrm{bau}}_{\mathrm{C}} -  \tau\left(\mathrm{E}^{\mathrm{bau}} - \tfrac{\tau}{2} \sum \varrho_i \right),
		      \\
		      \mathrm{E}^{\mathrm{\mathrm{mar}}}
		       & = \mathrm{E}^{\mathrm{bau}} - \sum P_i \varrho_i,
		       &
		      \mathcal{W}^{\mathrm{mar}}_{\mathrm{C}}
		       & = \mathcal{W}^{\mathrm{bau}}_{\mathrm{C}} -  \sum P_i \left(\mathrm{E}^{\mathrm{bau}}_i - \tfrac{P_i}{2} \varrho_i \right),
	      \end{align*}
	      where $P_i$ is a company-specific price per certificate.

	\item Obviously, under either scheme the overall emissions are reduced with respect to the business-as-usual case.
	      Under the tax scheme, the emission  factors  $\varrho_i = 1 / \pi_i^1 + 1 / \gamma_i$ for each company are uniformly scaled by the tax rate $\tau$.
	      In contrast, under the market scheme, $\varrho_i$ is weighted by the certificate price $P_i$ paid by the company $i$.
	\item Likewise, under either scheme, the overall wealth of companies is reduced relative to the business-as-usual case.
	      Under the tax scheme, each company's wealth reduction is determined by the emission factors $\varrho_i$, scaled quadratically by the tax rate $\tau$.
	      Under the market scheme, each company's wealth reduction is similarly determined by $\varrho_i$, weighted quadratically by its own certificate price $P_i$.
\end{itemize}

\begin{table}[htp]
	\centering
	\begin{threeparttable}
		\caption{Main notations used in the model.}
		\label{tab:notations}
		\begin{tabular}{@{}ll@{}}
			\toprule
			\multicolumn{2}{c}{\textbf{Basic model parameters}}  \\
			\midrule
			$\tau$                & Tax rate per ton of emissions under the tax scheme                                \\
			$\lambda$             & Penalty rate per ton of excess emissions under the market scheme                  \\
			$A$                   & Total number of certificates auctioned under the market scheme                    \\
			$\pi^0_i$             & Linear wealth coefficient of company $i$                                             \\
			$\pi^1_i$             & Quadratic wealth coefficient of company $i$                                          \\
			$\gamma_i$            & Emissions reduction cost parameter of company $i$                                    \\
			$\varrho_i$           & Emissions reduction factor $(1/\pi^1_i) + (1/\gamma_i)$  for company $i$             \\
			\midrule
			\multicolumn{2}{c}{\textbf{Equilibrium variables}$^\dagger$ \textbf{(baseline, tax, and market schemes)}} \\
			\midrule
			$q_i$                 & Quantity of goods produced by company $i$                                            \\
			$\mathrm{E}_i$        & Total emissions of company $i$                                                       \\
			$\mathcal{W}_{\mathrm{C}, i}$ & Wealth  of company $i$  \\
			\midrule
			\multicolumn{2}{c}{\textbf{Equilibrium variables}$^{\ddagger}$ \textbf{(tax and market schemes only)}}    \\
			\midrule
			$a_i$                 & Emissions reduction rate of company $i$                                              \\
			$\mathcal{W}_{\mathrm{R}}$    & Wealth of the regulatory authority                                                \\
			\midrule
			\multicolumn{2}{c}{\textbf{Equilibrium variables (market schemes only)}}                                  \\
			\midrule
			$S$                   & Spot price per unit of certificate                                                \\
			$P_i$                 & Price per unit of certificate for company $i$                                        \\
			$\delta_i$            & Quantity of certificates demanded by company $i$                                     \\
			$\mathcal{W}_{\mathrm{F}, i}$ & Wealth of company's $i$ financial intermediary (if applicable)                                 \\
			\bottomrule
		\end{tabular}
		\begin{tablenotes}
			\footnotesize
			\item Superscripts $\mathrm{bau}$, $\mathrm{tax}$, and $\mathrm{mar}$ are used for variables marked with $\dagger$ and $\ddagger$ to distinguish between the business-as-usual, tax, and market schemes cases, respectively. 
			\item Quantities without subscript $i$ (e.g.\ $\mathrm{E}^{\mathrm{mar}}$, $\mathcal{W}^{\mathrm{mar}}_{\mathrm{C}}$, $\mathcal{W}^{\mathrm{mar}}_{\mathrm{F}}$ denote economy-wide aggregates.
		\end{tablenotes}
	\end{threeparttable}
\end{table}

Under the market scheme, we still need to specify how the certificate prices $P_i$ paid by each company is formed.

\subsection{Market Scheme}
Under the market scheme, there are two possibilities for a company to purchase carbon certificates:
\begin{itemize}
	\item Either purchase directly at the regulatory authority's auction for an auction price $S$ referred to as \emph{spot price};
	\item Or purchase those certificates from financial intermediaries for an intermediated price $P_i$.
	      The financial intermediaries then participate in an auction in order to satisfy this demand by purchasing at the spot price $S$.
\end{itemize}
In the following we consider the two extreme cases, namely
\begin{itemize}
	\item All companies purchase certificates directly at the auction organized by the regulatory authority, referred to as \emph{spot market scheme};
	\item All companies purchase certificates  from intermediaries, referred to as \emph{purely intermediated  market scheme}.
	      The resulting wealth of the company's $i$ intermediary for a demand $\delta_i$ of certificates is therefore
	      \begin{equation*}
		      \delta_i \times (P_i - S),
	      \end{equation*}
	      where $P_i - S$ is the spread  between the intermediated and spot prices.
\end{itemize}

\begin{remark}
	In the main body of this work, we present the market scheme under two extreme cases: the spot market and the purely intermediated market scheme.
	In reality, companies that are tied to a financial intermediary can change to another financial intermediary or directly be present in the auction.
	This can  be modeled in the optimization problem of each company by adding opportunity costs (price elasticity) for switching to an alternative financial intermediary or participating in the auction.
\end{remark}

Each company has a certificate demand curve $\bm{\delta_i} = \bm{ \delta_i}(P_i) $, where we now make explicit the company's index $i$ and the dependence of $\bm{\delta_i}$ on the certificate price $P_i$ as specified in Proposition~\ref{prop:optimalquant}.
When certificates are purchased through a financial intermediary, this intermediary sets the price $\bm{P_i}$ to maximize its own wealth, i.e. 
\begin{equation*}
	\bm{P_i}(S)= \argmax_P \  \bm{ \delta_i}(P) (P-S),
\end{equation*}
where $S$ denotes the auction (spot) price of one certificate.
This induces an optimal demand  
$  \bm{ \delta_i}(\bm{P_i}(S))$ for the company, as characterized in Proposition~\ref{prop:optimalquant}.
Under direct auction procurement, where $\bm{P_i}(S) = S$, each company optimizes its demand directly as $\bm{ \delta_i}(S)$ for the given auction price.

\begin{proposition}\label{prop:quoted_price}
	Under the market schemes, for each company $i$:
	\begin{itemize}
		\item \textbf{Direct participation to the auction}:
		      When the company $i$ purchases certificates directly at a given spot price $S$, its optimal demand $\bm{ \delta_i}$ satisfies
		      \begin{equation*}
			      \bm{ \delta_i}(\bm{P_i}(S)) = \bm{\delta_i}(S)= \mathrm{E}^{\mathrm{bau}}_i - S \varrho_i.
		      \end{equation*}
		\item \textbf{Purchase through a financial intermediary}: When the company $i$ purchases certificates through a financial intermediary, under the feasibility assumption $0 < S \leq 2\lambda - \max\left( \mathrm{E}^{\mathrm{bau}}_i/ \varrho_i \right)$  on the given spot price $S$,  the intermediated price $\bm{P_i}(S)$ given $S$ satisfies
		      \begin{equation*}
			      \bm{P_i}(S)   = \frac{1}{2} \left(S + \frac{\mathrm{E}^{\mathrm{bau}}_i}{\varrho_i} \right).
		      \end{equation*}
		      For this price, the company's optimal certificate demand $\bm{ \delta_i}$ satisfies
		      \begin{equation*}
			      \bm{ \delta_i}(\bm{P_i}(S)) = \frac{1}{2} \left( \mathrm{E}^{\mathrm{bau}} _i- S \varrho_i \right).
		      \end{equation*}
	\end{itemize}
\end{proposition}
This follows as a special case of Proposition~\ref{prop:given_S}.

It is worth noting that, relative to direct market purchases, purchasing certificates through a financial intermediary results in a demand that is exactly half of the corresponding direct demand at the same spot price $S$.

\subsection{Regulatory Organized Auction: An Equilibrium Price}
With the optimal demand of certificates $\bm{\delta_i}(\bm{P_i}(S))$ on the auction market---either directly for a company $i$ participating in the Radner equilibrium type auction (see \citep{radner1972}), or indirectly through intermediary---the resulting auction price $\bm{S}$ is determined by clearing, i.e.
\begin{equation*}
	A=\sum \bm{\delta_i}(\bm{P_i}(\bm{S})),
\end{equation*}
where $A$ denotes the total number of certificates auctioned by the regulatory authority, meant as the overall target for carbon emissions.

As a special case of Proposition~\ref{prop:given_S}, we obtain:
\begin{proposition}\label{prop:price_exchange}
	Let $\varrho := \sum \varrho_i$. 
	\begin{itemize}[fullwidth]
		\item \textbf{Spot market scheme:}
		      Under the feasibility condition $\mathrm{E}^{\mathrm{bau}} - \lambda \varrho \leq A  < \mathrm{E}^{\mathrm{bau}}$ on $A$, the equilibrium spot price $\bm{S}$ on the market satisfies
		      \begin{equation*}
			      \bm{S} =  \frac{1}{\varrho}\left(\mathrm{E}^{\mathrm{bau}}-A\right).
		      \end{equation*}
		\item \textbf{Purely intermediated market scheme:}
		      Under the feasibility condition
		      \begin{equation*}
			      \left[\mathrm{E}^{\mathrm{bau}} + \varrho[ \max(\mathrm{E}^{\mathrm{bau}}_i/\varrho_i ) -2\lambda ]  \right]/2 \leq A  <  \sum \left[ (\mathrm{E}^{\mathrm{bau}}_i/2) \vee ( \mathrm{E}^{\mathrm{bau}}_i- \lambda \varrho_i ) \right]
		      \end{equation*}
		      on $A$, the equilibrium spot price $\bm{S}$ on the market satisfies
		      \begin{equation*}
			      \bm{S} =  \frac{1}{\varrho}\left(\mathrm{E}^{\mathrm{bau}} - 2A\right).
		      \end{equation*}
	\end{itemize}
\end{proposition}

Obviously, Proposition~\ref{prop:price_exchange} shows that the spot price $\bm{S}$ is  lower under the intermediated market scheme than  under the spot market scheme.

\section{Comparison of the Tax and Market Schemes }\label{sec:compare}
To compare the different carbon pricing schemes, we assume that the regulator sets the tax rate $\tau$, penalty $\lambda$ and number of certificates $A$ such that the resulting aggregated emissions are identical under the tax and market schemes, i.e.  $\mathrm{E}^{\mathrm{tax}} = \mathrm{E}^{\mathrm{mar}}$.
By Proposition~\ref{prop:optimalquant}, the certificates demand $\bm{ \delta_i} = \bm{\delta_i}(\bm{P_i}(\bm{S}))$ of each company $i$ under the market schemes coincides with its total emissions, i.e.  $\bm{ \delta_i} = \mathrm{E}^{\mathrm{mar}}_i$.
It then follows that  aggregate emissions under the market schemes satisfy $\mathrm{E}^{\mathrm{mar}} = \sum \bm{\delta_i} = A$.
Hence, the tax rate $\tau$ must be chosen such that aggregate emissions under the tax scheme also satisfy $\mathrm{E}^{\mathrm{tax}} = A$.
Turning to the wealth generated by the regulatory authority, we obtain:
\begin{equation*}
	\mathcal{W}^{\mathrm{tax}}_{\textrm{R}}  = \tau A,
	\quad
	\mathcal{W}^{\mathrm{mar}}_{\textrm{R}}  = \lambda \sum  \left( \mathrm{E}^{\mathrm{mar}}_i -\bm{\delta_i} \right)^+  + \bm{S} A= \bm{S} A.
\end{equation*}
The second equality in the market case follows from the identity $\mathrm{E}^{\mathrm{mar}}_i = \bm{\delta_i}$, as established in Proposition~\ref{prop:optimalquant}, which implies that the penalty term vanishes.
The following result presents the wealth outcomes for companies and the regulator under both the tax and market schemes, including the wealth earned by financial intermediaries under the market scheme.

\begin{theorem}\label{thm:comparison}
	Assume that the tax rate $\tau$ and  penalty rate $\lambda$ are chosen such that the aggregate emissions are identical across all policy schemes, i.e. equal to $A$ (as they have to be under the market scheme).
	Then
	\begin{itemize}[fullwidth]
		\item \textbf{Tax scheme}:
		      The tax rate is given by $ \tau =  (\mathrm{E}^{\mathrm{bau}}-A)/\varrho$ and the wealth of company $i$ is
		      \begin{equation*}
			      \mathcal{W}^{\mathrm{tax}}_{\mathrm{C}, i} = \mathcal{W}^{\mathrm{bau}}_{\mathrm{C}, i} - \tau \left(\mathrm{E}^{\mathrm{bau}}_i - \frac{\tau}{2}\varrho_i \right).
		      \end{equation*}
		      This leads to the aggregate wealth of the companies and the wealth of the regulator, given by
		      \begin{equation*}
			      \mathcal{W}^{\mathrm{tax}}_{\mathrm{C}}  = \mathcal{W}^{\mathrm{bau}}_{\mathrm{C}} - \tau \left(\mathrm{E}^{\mathrm{bau}} - \frac{\tau}{2}\varrho \right),
			      \quad
			      \mathcal{W}^{\mathrm{tax}}_{\mathrm{R}}  = \tau A.
		      \end{equation*}
		\item \textbf{Spot market scheme}:
		      Outcomes are identical to those under the tax scheme.
		      Specifically,
		      \begin{equation*}
			      \mathcal{W}^{\mathrm{mar}}_{\mathrm{C}, i} = \mathcal{W}^{\mathrm{tax}}_{\mathrm{C}, i}
		      \end{equation*}
		      and the aggregate wealth of the companies and the wealth of the regulator  satisfy
		      \begin{equation*}
			      \mathcal{W}^{\mathrm{mar}}_{\mathrm{C}} = \mathcal{W}^{\mathrm{tax}}_{\mathrm{C}},
			      \quad
			      \mathcal{W}^{\mathrm{mar}}_{\mathrm{R}}  = \mathcal{W}^{\mathrm{tax}}_{\mathrm{R}}.
		      \end{equation*}
		\item \textbf{Purely intermediated market scheme}:
		      The optimal intermediated price faced by the company $i$ satisfies $\bm{P_i}(\bm{S}) =  \tfrac{\tau}{2} + \tfrac{1}{2}(\mathrm{E}^{\mathrm{bau}}_i / \varrho_i- A/\varrho)$, the  wealth of the financial intermediary serving company $i$ and   of the company $i$  are  given by
		      \begin{align*}
			      \mathcal{W}^{\mathrm{mar}}_{\mathrm{F}, i}
			       & = \frac{(\mathrm{E}_i^{\mathrm{mar}})^2}{\varrho_i}
			      =  \dfrac{\varrho_i}{4}\displaystyle \left( \frac{\mathrm{E}^{\mathrm{bau}}_i }{ \varrho_i}  -\frac{\mathrm{E}^{\mathrm{bau}} }{ \varrho} +  \frac{2A }{ \varrho } \right)^2
			      ,
			      \\
			      \mathcal{W}^{\mathrm{mar}}_{\mathrm{C}, i}
			       & = \displaystyle \mathcal{W}^{\mathrm{tax}}_{\mathrm{C}, i} - \frac{3}{2}\mathcal{W}^{\mathrm{mar}}_{\mathrm{F}, i}  + \frac{ A\varrho_i}{\varrho}  \left( \frac{\mathrm{E}^{\mathrm{bau}}_i }{ \varrho_i}  -\frac{\mathrm{E}^{\mathrm{bau}} }{ \varrho} +  \frac{3A }{ 2\varrho } \right).
		      \end{align*}
		      This leads to  a total wealth for the financial intermediaries, the companies  and the regulator  given by
		      \begin{align*}
			      \mathcal{W}^{\mathrm{mar}}_{\mathrm{F}} & = \sum \frac{(\mathrm{E}_i^{\mathrm{mar}})^2}{\varrho_i} ,
			                                                &
			      \mathcal{W}^{\mathrm{mar}}_{\mathrm{C}} & =\mathcal{W}^{\mathrm{tax}}_{\mathrm{C}} - \frac{3}{2}  \mathcal{W}^{\mathrm{mar}}_{\mathrm{F}} +\frac{3}{2 \varrho}A^2,\\
			      \mathcal{W}^{\mathrm{mar}}_{\mathrm{R}}      & = \mathcal{W}^{\mathrm{tax}}_{\mathrm{R}} - \frac{A^2}{\varrho}.
		      \end{align*}
	\end{itemize}
\end{theorem}

This follows as special cases of Theorem~\ref{thm:comparison_hybrid}.
In view of these results:
\begin{itemize}[fullwidth]
	\item In a perfect competition case where all the companies  are active on the auction market, there is no difference between the tax and the market scheme.
	\item The regulatory authority collects  more wealth under the tax scheme than under market schemes.
	\item In the purely intermediated market scheme, the aggregate wealth of the companies is lower than under the tax scheme, namely
	      \begin{equation*}
		      \mathcal{W}^{\mathrm{mar}}_{\mathrm{C}}
		      = \mathcal{W}^{\mathrm{tax}}_{\mathrm{C}} + \tfrac{3}{2}\underbrace{\left(\tfrac{A^2}{\varrho} - \sum \tfrac{(\mathrm{E}^{\mathrm{mar}}_i)^2}{\varrho_i}\right)}_{\leq 0}
		      \leq \mathcal{W}^{\mathrm{tax}}_{\mathrm{C}},
	      \end{equation*}
	      where the inequality  follows from Jensen's inequality.
	      The corresponding inequality, however, does not necessarily hold at the level of each  individual company.
	      That is, the wealth of a given company under the market scheme may be either higher or lower than under the tax scheme, depending on its emission reduction factor.
	      More precisely, company $i$ benefits from the market scheme relatively to the tax scheme if its carbon emission, adjusted by its technological and production factor $\varrho_i$, exceed one fourth of the corresponding economy-wide average, i.e.
	      \begin{equation*}
		      \frac{\mathrm{E}^{\mathrm{mar}}_i}{\varrho_i}  >  \frac{A}{4 \varrho}.
	      \end{equation*}
	      See Section~\ref{sec:numeric} for numerical illustrations.
\end{itemize}

The following corollary distills the main implications of Theorem~\ref{thm:comparison} in terms of aggregate economic output (GDP). 
We define the GDP of the real economy as the total wealth, measured as the sum of the wealth of companies, financial intermediaries, and the regulator: 
\begin{equation*}
	\mathrm{GDP}^{\cdot} = \mathcal{W}^{\cdot}_{\mathrm{C}} + \mathcal{W}^{\cdot}_{\mathrm{F}} + \mathcal{W}^{\cdot}_{\mathrm{R}},
\end{equation*}
where $\mathcal{W}^{\bm{\mathrm{\cdot}}}_{\mathrm{C}}$,  $\mathcal{W}_{\mathrm{F}}^{\bm{\mathrm{\cdot}}}$ and $\mathcal{W}^{\bm{\mathrm{\cdot}}}_{\mathrm{R}}$  denote the aggregate wealth of the companies, of the regulatory authority, and of the financial intermediaries under the market ($\cdot =\mathrm{mar}$) or tax scheme ($\cdot=\mathrm{tax}$).

We obtain as a special case of Corollary~\ref{cor:jensen_hybrid}:
\begin{corollary}\label{cor:jensen}
	Under the purely intermediated market scheme, the GDP of the real economy satisfies
	\begin{align*}
		\mathrm{GDP}^{\mathrm{mar}}
		 & = \mathrm{GDP}^{\mathrm{tax}} + \tfrac{1}{2}\underbrace{\left(  \tfrac{A^2}{\varrho}-\sum \tfrac{(\mathrm{E}^{\mathrm{mar}}_i)^2}{\varrho_i} \right)}_{\leq 0}.
	\end{align*}
	Moreover, the combined wealth of the companies and of the regulator is lower under the market scheme than under the tax scheme, i.e.
	\begin{align*}
			 	& \mathcal{W}^{\mathrm{mar}}_{\mathrm{C}} + \mathcal{W}^{\mathrm{mar}}_{\mathrm{R}}
			      =  \mathcal{W}^{\mathrm{tax}}_{\mathrm{C}} + \mathcal{W}^{\mathrm{tax}}_{\mathrm{R}}+  \underbrace{ \tfrac{3}{2} \left(  \tfrac{A^2}{\varrho}-\sum \tfrac{(\mathrm{E}^{\mathrm{mar}}_i)^2}{\varrho_i} \right) - \frac{A^2}{\varrho} }_{<0}.
	\end{align*}
\end{corollary}

According to the above corollary, the aggregate wealth of intermediaries, $\mathcal{W}^{\mathrm{mar}}_\mathrm{F}$, represents a diversion of value away from companies and the regulator.

\begin{example}[Identical Companies]\label{eg:symetric}
	If  all companies are identical, i.e. all the company-specific parameters share the same values,  then under a purely intermediated market scheme,  Theorem~\ref{thm:comparison} implies
	\begin{equation*}
		\bm{S}       =  \tau - \frac{A}{N \varrho_1 },
		\quad
		\bm{P_i}(\bm{S})  = \tau = \frac{\mathrm{E}^{\mathrm{bau}}}{N \varrho_1} - \frac{A}{N \varrho_1}, \quad \text{with} \quad \varrho_i =\varrho_1, 
	\end{equation*}
	and
	\begin{align*}
		& \mathcal{W}^{\mathrm{mar}}_{\mathrm{C}, i}  =  \mathcal{W}^{\mathrm{tax}}_{\mathrm{C}, i} = \mathcal{W}^{\mathrm{tax}}_{\mathrm{C}, 1}, \quad 
		\qquad  
		\mathcal{W}^{\mathrm{mar}}_{\mathrm{F}, i}    = \displaystyle \frac{A }{N^2 \varrho_1 },
		\\
		& \mathcal{W}^{\mathrm{mar}}_\mathrm{F}  = N  \mathcal{W}^{\mathrm{mar}}_{\mathrm{F}, 1}, 
		\qquad 
		\mathcal{W}^{\mathrm{mar}}_\mathrm{R}  = \mathcal{W}^{\mathrm{tax}}_\mathrm{R}  -  \mathcal{W}^{\mathrm{mar}}_\mathrm{F},
	\end{align*}
	with resulting wealth of the companies and of the regulator such that
	\begin{equation*}
		\mathcal{W}^{\mathrm{mar}}
		  = \mathcal{W}^{\mathrm{tax}} - \mathcal{W}^{\mathrm{mar}}_\mathrm{F} < \mathcal{W}^{\mathrm{tax}}.
	\end{equation*}
	In this symmetric setup, each company faces the same carbon price $\bm{P_i} = \tau$, making the economic outcomes under tax and market scheme identical from the companies' perspective.
	However, under the purely intermediated market scheme, the financial intermediaries capture a portion of the wealth that would accrue entirely to the regulator under the tax scheme.
\end{example}

\section{Numerical Case Study and Discussion}\label{sec:numeric}
Since the tax and spot market schemes yield identical outcomes, we use a simple numerical case study to illustrate how the purely intermediated market scheme diverges from the tax scheme in terms of both wealth and emissions.
Note that the wealth of the green technology provider associated with company $i$ corresponds  to $ c_i\left(\bm{ a^{\mathrm{tax}} }_i, \bm{ q^{\mathrm{tax}} }_i \right)$ under the tax scheme and $c_i\left(\bm{ a^{\mathrm{mar} } }_i, \bm{ q^{\mathrm{mar}} }_i \right)$ under the market scheme. 
	In order to compare the level of investment in the green transition under each policy scheme, we include the wealth of tech providers
     $\sum c_i   \left(\bm{ a^{\mathrm{tax}} }_i, \bm{ q^{\mathrm{tax}} }_i \right)$ vs.\ $\sum c_i \left(\bm{ a^{\mathrm{mar} } }_i, q^{\mathrm{mar}}_i \right)$
    in our numerical analysis.
\begin{figure}[htp!]
	\centering
	\includegraphics[width=\textwidth]{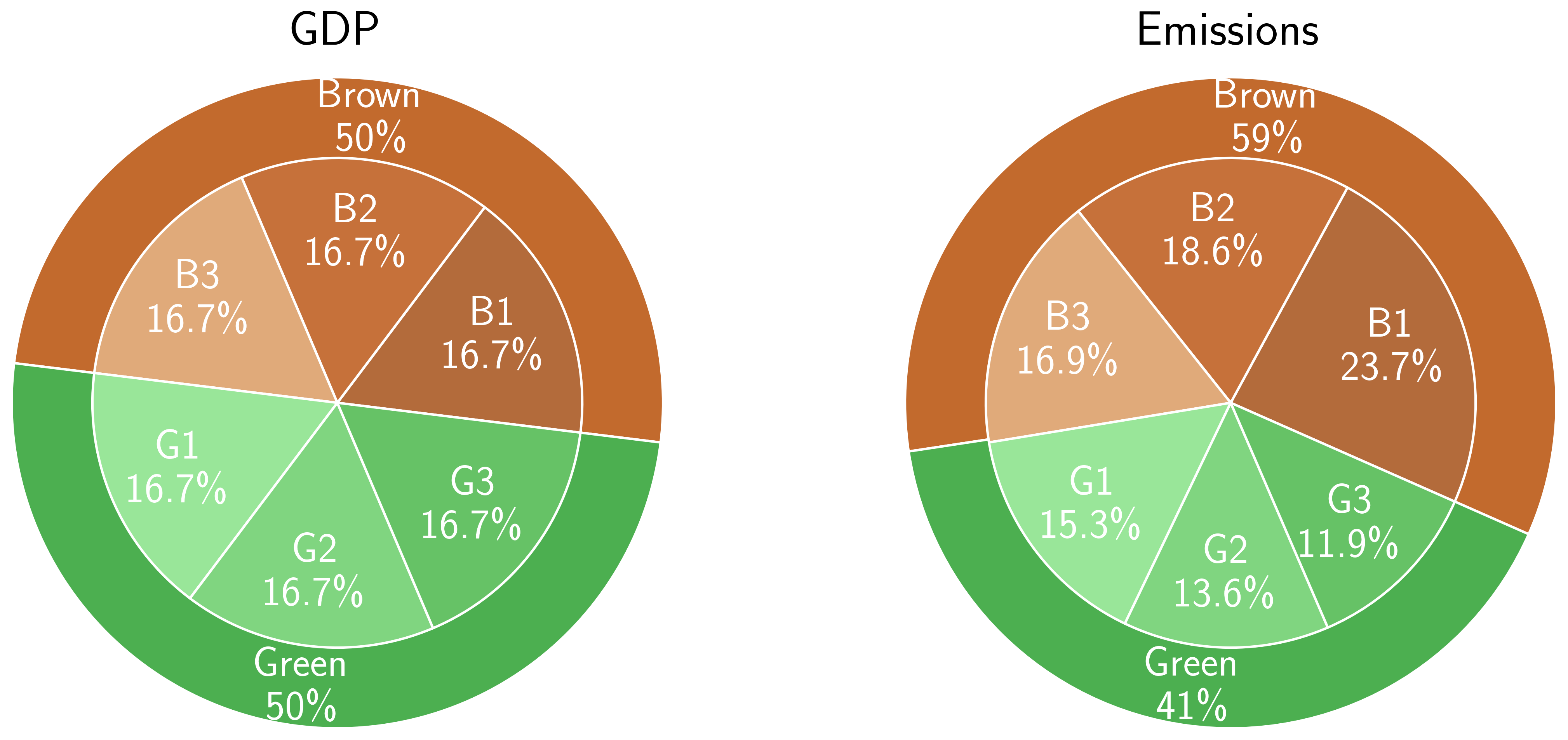}
	\caption{Business-as-usual GDP shares $\left(\mathcal{W}^{\mathrm{bau}}_{\mathrm{C}, i}/ \mathcal{W}^{\mathrm{bau}}_{\mathrm{C}}\right)\!100$ and emission shares $\left(\mathrm{E}^{\mathrm{bau}}_{i}/ \mathrm{E}^{\mathrm{bau}}\right)\!100$ for each company $i$.}
	\label{fig:pre_carbon_plot}
\end{figure}

Throughout the analysis, we consider $6$ groups of companies, each contributing equally to the GDP in the  business-as-usual  case.
However, these groups differ along two key dimensions:
\begin{itemize}
	\item \textbf{Carbon emission intensity}:
	      Companies are categorized according to their carbon intensity from the most carbon intensive (\emph{brown}, categories B1 to B3) to the least intensive (\emph{green}, categories G1 to G3).
	      Formally, the carbon intensity of company $i$ is defined as $\kappa_i = \mathrm{E}^{\mathrm{bau}}_i / \mathcal{W}^{\mathrm{bau}}_{\mathrm{C}, i}$ representing the amount of carbon per monetary unit of GDP contribution.
	\item \textbf{Green costs}:
	      The technological costs of reducing emissions,  $\gamma_i$, are assumed to be disconnected from  companies' carbon intensities.
	      That is, a carbon intensive (brown) company may face relatively low abatement costs, while a less carbon intensive (green) company might incur high marginal costs to further reduce an already minimal carbon footprint.
	      To reflect this heterogeneity, we assign relatively high abatement cost parameters $\gamma_i$ to companies in categories B2 and G2.
\end{itemize}

The overall setting is illustrated in Figure~\ref{fig:pre_carbon_plot}, while the exact parameters of the model are displayed in Table~\ref{tab:parameters} as to align broadly to the data from the EU Emissions Trading System (ETS) viewer\footnote{
	See \url{https://www.eea.europa.eu/en.}
}.
Since our model is static and simplified, we consider a targeted emission reductions of $40\%$, coinciding roughly with the amount of certificates issued by the EU recently, as reported in \citet{EEX2025}.\footnote{Since the number of certificates $A$ is fixed and coincides with the overall emissions in the economy after abatement, by Proposition~\ref{prop:optimalquant}, we can derive the tax rate $\tau$ as a function of $A$.}
\begin{table}[htp!]
    \centering
    \caption{Parameterization of model inputs.
        The linear coefficients $\pi^0_i$ of each firm's profit function are expressed in euros.
        The technology cost factors $\gamma_i$ are expressed in €$/\text{ton}^2$, while the quadratic coefficients $\pi^1_i$ are given in €$/\text{ton}^2$; both quantities are reported in units of $10^{-6}$.}
    \label{tab:parameters}
    \begin{tabular}{ @{} l rrr @{} }
      \toprule
      & \multicolumn{3}{c}{Parameters} \\
      \cmidrule(lr){2-4}
      \multicolumn{1}{c}{Group} 
      & \multicolumn{1}{c}{$\pi^0_i$}
      & \multicolumn{1}{c}{$\pi^1_i$}
      & \multicolumn{1}{c}{$\gamma_i$} \\
      \midrule
    \quad \quad B1 & 475.65 & 1.13 & 10 \\
    \quad \quad B2 & 605.38 & 1.83 & 30 \\
    \quad \quad B3 & 665.91 & 2.22 & 5 \\
    \quad \quad G1 & 739.90 & 2.74 & 10 \\
    \quad \quad G2 & 832.39 & 3.46 & 30 \\
    \quad \quad G3 & 951.31 & 4.52 & 5 \\
      \bottomrule
    \end{tabular}
\end{table}

\subsection{Wealth Impact}
Table~\ref{tab:value} and Figure~\ref{fig:gdp_tax_vs_fi} respectively compare the GDP contributions (each economic agent's wealth) of all actors and their relative shares in the business-as-usual case and under  the tax and purely intermediated market schemes.
We can draw the following conclusions:
\begin{itemize}
	\item \textbf{Tax scheme}: The aggregate GDP contribution of all actors, including those supplying green technologies, is approximately 22\% lower than in the business-as-usual case, decreasing from €600 billion to €467 billion (see Table~\ref{tab:value}).
		      As illustrated in Figure~\ref{fig:gdp_tax_vs_fi}, while the regulator recaptures a significant share of the GDP under the tax scheme, green companies become the main GDP contributors relative to brown companies.
	\item \textbf{Tax vs.\ market scheme}: As shown in Table~\ref{tab:value}, a purely intermediated market scheme reduces the overall GDP contribution with respect to the business-as-usual case by $23\%$, decreasing from €600 billion to €462 billion.
	      However, the GDP shares among the actors shift drastically.
	      Financial intermediaries collect, mainly at the expense of the regulator, a large share of the GDP.
	      Figure~\ref{fig:gdp_tax_vs_fi} suggests that the relative GDP share of brown companies remains roughly the same in the tax and market scheme.
	      Green companies are relatively better off compared to brown companies under both schemes, see Table~\ref{tab:value}.
	      As visible in the table, regardless of their carbon intensity, most of the brown companies are more profitable under the market scheme,  while green ones are more  profitable under the tax scheme.
\end{itemize}
\begin{figure}[htp!]
	\centering
	\includegraphics[width=1.0\textwidth]{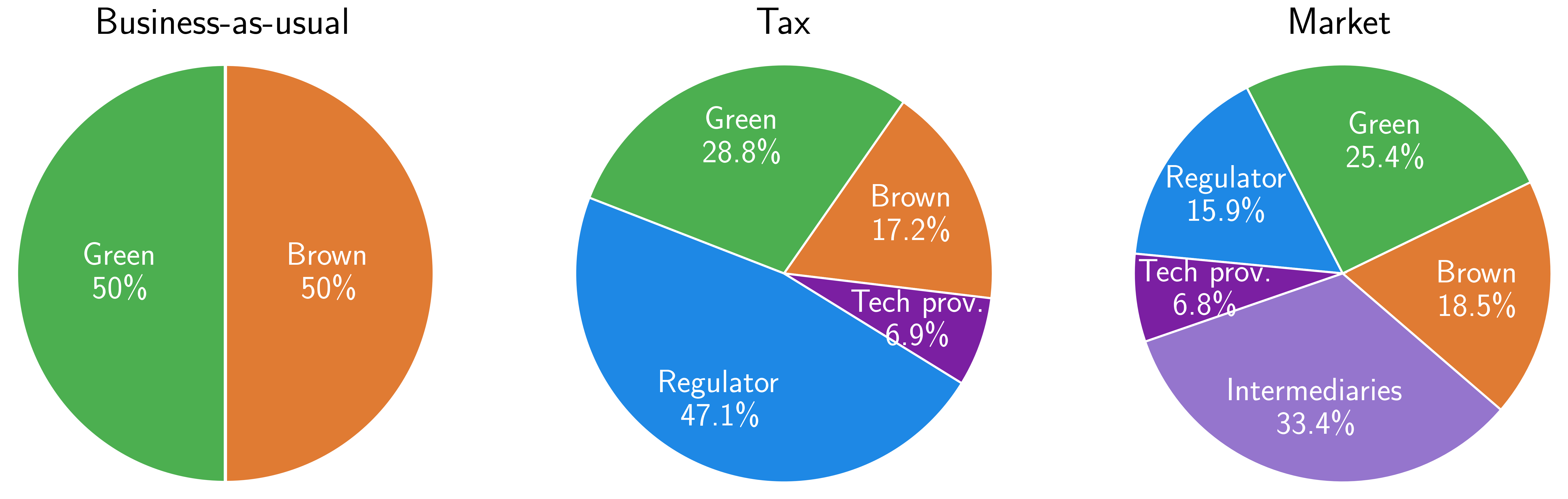}
	\caption{The GDP shares of the different actors in the business-as-usual case and under the tax and purely intermediated market schemes.}
	\label{fig:gdp_tax_vs_fi}
\end{figure}
\begin{table}[htp!]
    \centering
    \begin{threeparttable}
        \caption{GDP and emission contributions of  each actor in the business-as-usual case and under  the tax and market schemes.}
        \label{tab:value}
        \begin{tabular}{@{}lrrrcrrr@{}}
            \toprule
            & \multicolumn{3}{c}{ \makecell{ GDP contributions\\ (billion  €) } } & & \multicolumn{3}{c}{ \makecell{Emission contributions \\ (million tons)} } \\
            \cmidrule{2-4} \cmidrule{6-8}
            Group & $\bm{\mathrm{bau}}$ & $\bm{\mathrm{tax}}$ & $\bm{\mathrm{mar}}$ & & $\bm{\mathrm{bau}}$ & $\bm{\mathrm{tax}}$ & $\bm{\mathrm{mar}}$ \\
            \midrule
            Brown & 300 & 80 & 86 & & 1,051 & 363 & 411 \\
            \quad B1 & 100 & 17 & 23 & & 420 & 115 & 159 \\
            \quad B2 & 100 & 25 & 22 & & 330 & 150 & 135 \\
            \quad B3 & 100 & 38 & 41 & & 300 & 98 & 116 \\
            Green & 300 & 134 & 117 & & 721 & 345 & 298 \\
            \quad G1 & 100 & 38 & 35 & & 270 & 126 & 111 \\
            \quad G2 & 100 & 41 & 27 & & 240 & 140 & 103 \\
            \quad G3 & 100 & 55 & 56 & & 210 & 79 & 83 \\
            Regulator & --- & 220 & 73 & & --- & --- & --- \\
            Tech providers & --- & 32 & 31 & & --- & --- & --- \\
            Financial intermediaries & --- & --- & 154 & & --- & ---& --- \\
            \midrule
            \textbf{Total} & 600 & 467 & 462 & & 1,772 & 709 & 709 \\
            \bottomrule
            \end{tabular}
        \end{threeparttable}
    \end{table}

\subsection{Emission Contribution Impact}
Table~\ref{tab:value} and Figure~\ref{fig:gdp_tax_vs_fi} also compare, respectively, the emission contributions of compliant companies and their relative shares in the business-as-usual case and under the tax and purely intermediated market schemes.
Under our comparison framework, tax and market schemes achieve the same total emissions.
We can draw the following conclusions:
\begin{itemize}
	\item \textbf{Tax scheme}: As shown in Table~\ref{tab:value} and Figure~\ref{fig:emission_tax_vs_fi}, under the tax scheme, brown companies reduce their carbon emissions by a greater amount than green companies, with emissions declining from 1,051 to 363 million tons for brown companies and from 721 to 345 million tons for green companies.
	Brown companies with lower emission reduction  costs (i.e. B1 and B3) do reduce their emissions more than the others, see Table~\ref{tab:value}.
	\item \textbf{Tax vs.\ market scheme}: Under the purely intermediated market scheme, the emission shares of each category (Green vs.\ Brown)  are similar as in the  business-as-usual  case, see  Figure~\ref{fig:emission_tax_vs_fi}.
	      The table and figure also show that the emission values (or shares) of brown companies are substantially higher than those of green companies.
	      In other words, under the purely intermediated market scheme, brown companies are less incentivized than green companies to reduce their emissions; by contrast, under the tax scheme, brown companies emerge as the primary drivers of emission reduction.
\end{itemize}
\begin{figure}[htp!]
	\centering
	\includegraphics[width=1.0\textwidth]{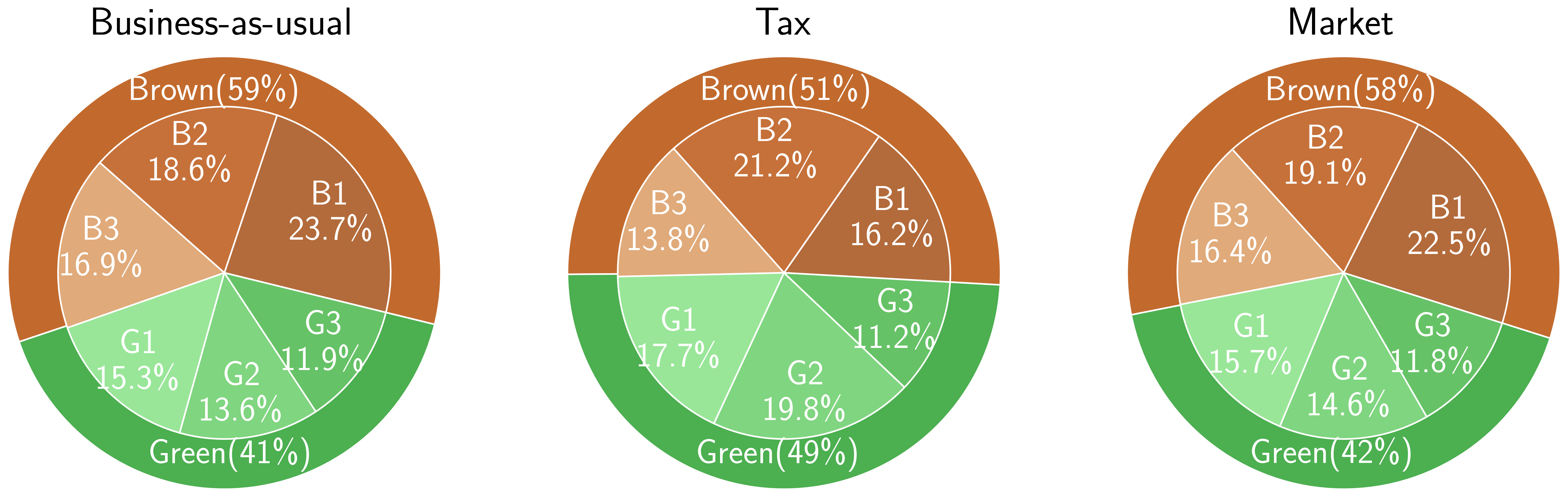}
	\caption{The emission shares of companies in the business-as-usual case and under the tax and  purely intermediated market schemes.}
	\label{fig:emission_tax_vs_fi}
\end{figure}

\section{Carbon pricing under emissions' uncertainty}\label{sec:princing_random}

In this section, we assume that emissions during the production period are stochastic rather than deterministic.
This is notably the case regarding energy providers, for which operational considerations and resource availability are highly contingent on weather conditions.
This  uncertainty  influences the demand $\delta$ of the companies for certificates, and therefore the resulting equilibrium.

To model this uncertainty, we assume that instead of a fixed unit of carbon per unit of production, the reference level is now a positive random variable $X$  with a continuous distribution and expectation $\mathbb{E}[X]=1$.
The company's wealth function will accordingly be modified into
\begin{align}
	\bm{\mathrm{tax}}: &  &  & \pi(q) - \mathbb{E} \left[ \frac{\gamma}{2} \left( (1 - e^{-a})  q X  \right)^2 \right] - \tau \mathbb{E}\left[ e^{-a}  q  X \right]
	\notag                                                                                                                    \\
	                   &  &  & =\pi(q) - \underbrace{\frac{\gamma \sigma^2}{2} q^2 (1 - e^{-a})^2}_{=c(q,a)} - \tau e^{-a}  q  \label{eq:comp_wealth_taxR}                                                  \\
	\bm{\mathrm{mar}}: &  &  & \pi(q) - \mathbb{E} \left[ \frac{\gamma}{2} \left( (1 - e^{-a}) q X  \right)^2 \right] - \delta P- \lambda\, \mathbb{E} \left[ \left(e^{-a}  q  X - \delta \right)^+ \right]
	\notag                                                                                                                                                             \\
	                   &  &  & =  \pi(q) -c(q, a) - \delta P - \lambda\, \mathbb{E} \left[ \left(e^{-a}  q  X - \delta \right)^+ \right],\notag                        
\end{align}
where $\sigma^2 =\mathbb{E}[X^2] = 1 + \mathrm{Var}(X)$.
From this, it turns out that the cost $c(q,a)$ is similar as in the deterministic case, but scaled by a factor $\sigma^2$ related to the variance of $X$.
Hence, the tax objective function of the companies does not change much.
As for the market case, the penalization term similar to a call option price, will shift the optimal values of the resulting optimization problem by factors related to the expected shortfall and value-at-risk of $X$.
Recall that the value-at-risk and expected shortfall of $X$ at confidence level $0<\epsilon \leq 1$ satisfy\footnote{For a more detailed discussion of the properties of these risk measures and the connections between them, we refer the reader to  Example 2.3 of \citet{bental2007}.
}
\begin{equation}
	\begin{split}
		\mathbb{V}{\rm a}\mathbb{R}_{\epsilon}(X) & := \inf\{ y \in \mathbb{R} \colon \mathbb{P}[X \leq y] \geq 1- \epsilon \},
		\\
		\mathbb{ES}_{\epsilon}(X)  & := \frac{1}{\epsilon}\int_0^\epsilon \mathbb{V}{\rm a}\mathbb{R}_u(X) du                                                                                  \\
		                         & = \inf_{m \in \mathbb{R} } \left\{ m + \tfrac{1}{\epsilon}\mathbb{E}\left[ (X - m)^+ \right] \right\}                                     \\
		                         & = \mathbb{V}{\rm a}\mathbb{R}_{\epsilon}(X) + \tfrac{1}{\epsilon} \mathbb{E}\left[ \left(X - \mathbb{V}{\rm a}\mathbb{R}_{\epsilon}(X)\right)^+ \right].\label{eq:oce_conn}
	\end{split}
\end{equation}

\begin{remark}
	Under the market scheme, it is natural to present the following proposition in terms of expected shortfall and value-at-risk of $X$.
	In particular, the optimization problem with respect to the certificates demand $\delta$ can be reformulated as the value of a call option, i.e. minimizing over $\delta$
	\begin{equation*}
		\delta P + \lambda \mathbb{E} \left[ \left(e^{-a}  q  X - \delta \right)^+ \right] = \delta P + \lambda e^{-a} q\ \mathbb{E}\left[ \left(X - \frac{\delta}{e^{-a} q} \right)^+ \right]  = \delta P + \lambda e^{-a} q\ \mathbb{C},
	\end{equation*}
	where $\mathbb{C} = \mathbb{E}\left[ \left(X - \frac{\delta}{e^{-a} q} \right)^+ \right] $ is the standardized price of a call option on $X$ and strike price $\delta / (e^{-a} q)$.
	Hence, first order conditions are related to the ``delta'' of this call option.
\end{remark}
The following proposition mirrors the results of Proposition~\ref{prop:optimalquant} in the case of random carbon emissions where we adjust $\varrho$ by the variance factor
\begin{equation*}
	\hat{\varrho} \colon  = \frac{1}{\pi^1} + \frac{1}{\gamma \sigma^2} < \varrho.
\end{equation*}

\begin{proposition}\label{prop:optimalquant_random}
	Let $0 < P < \lambda$ be given, representing the unit certificate price faced by the company.
	Under the tax scheme $(\mathrm{tax})$ and the market schemes $(\mathrm{mar})$, it holds
	\begin{itemize}[fullwidth]
		\item $\bm{\mathrm{tax}}$: the optimal production $\bm{ q^{\mathrm{tax}} }$ and  emissions reduction rate $\bm{ a^{\mathrm{tax}} }$ are given by
		      \begin{equation*}
			      \bm{ q^{\mathrm{tax}} }     = \bm{ q^{\mathrm{bau}}  }- \tau\frac{1}{\pi^1},
			      \qquad
			      e^{-\bm{ a^{\mathrm{tax}} }}  = \displaystyle \frac{\mathrm{E}^{\mathrm{bau}} - \tau \hat{\varrho}}{\bm{ q^{\mathrm{tax}} }},
		      \end{equation*}
		      with resulting (expected)  carbon emissions  $\mathrm{E}^{\mathrm{tax}}$ and wealth $\mathcal{W}^{\mathrm{tax}}_{\mathrm{C}}$
		      \begin{equation*}
			      \mathrm{E}^{\mathrm{tax}}       = \mathrm{E}^{\mathrm{bau}} - \tau \hat{\varrho},
			      \qquad
			      \mathcal{W}^{\mathrm{tax}}_{\mathrm{C}}  = \mathcal{W}^{\mathrm{bau}}_{\mathrm{C}} - \tau \left(\mathrm{E}^{\mathrm{bau}} - \frac{\tau}{2} \hat{\varrho}\right).
		      \end{equation*}

		\item $\bm{\mathrm{mar}}$: the optimal production $\bm{ q^{\mathrm{mar}} }$, emissions reduction rate $\bm{ a^{\mathrm{mar} } }$, and certificate demand $\delta$ are given by
		      \begin{align*}
			      \bm{ q^{\mathrm{mar}} }       & = \displaystyle \bm{ q^{\mathrm{bau}}  }- P \, \mathbb{ES}_{P/\lambda}(X) \frac{1 }{\pi^1},
			      \\
			      e^{-\bm{ a^{\mathrm{mar} } }} & = \displaystyle \frac{ \mathrm{E}^{\mathrm{bau}} - P\, \mathbb{ES}_{P/\lambda}(X) \hat{\varrho}}{ \bm{ q^{\mathrm{mar}} }},             \\
			      \bm{\delta}                      & =  \displaystyle \mathbb{V}{\rm a}\mathbb{R}_{P/\lambda}(X) \left[ \mathrm{E}^{\mathrm{bau}} -  P\, \mathbb{ES}_{P/\lambda}(X) \hat{\varrho} \right],
		      \end{align*}
		      with resulting (expected) carbon emissions $\mathrm{E}^{\mathrm{mar} }$ and wealth $\mathcal{W}^{\mathrm{mar}}_{\mathrm{C}}$
		      \begin{align*}
			      \mathrm{E}^{\mathrm{mar} }      & = \mathrm{E}^{\mathrm{bau}} - P\, \mathbb{ES}_{P/\lambda}(X) \hat{\varrho},\\
			      \mathcal{W}^{\mathrm{mar}}_{\mathrm{C}} & = \mathcal{W}^{\mathrm{bau}}_{\mathrm{C}} - P\,  \mathbb{ES}_{P/\lambda}(X) \left(\mathrm{E}^{\mathrm{bau}} - \frac{P\, \mathbb{ES}_{P/\lambda}(X)}{2} \hat{\varrho} \right).
		      \end{align*}
	\end{itemize}
\end{proposition}

\begin{proof}
	Under the tax scheme, the company's expected wealth function~\eqref{eq:comp_wealth_taxR} is
	\begin{align*}
		f(q,a) & =  \pi^0 q - \frac{\pi^1}{2} q^2 - \frac{\gamma \sigma^2 }{2} \big((1-e^{-a})q\big)^2 - \tau e^{-a} q                                                  \\
		       & =\pi^0 q - \frac{\pi^1}{2} q^2 - \frac{\hat{\gamma} }{2} \big((1-e^{-a})q\big)^2 - \tau e^{-a} q, \quad \text{with } \hat{\gamma} : = \gamma \sigma^2.
	\end{align*}
	This is the same wealth function as in the deterministic case, except for the modified  green cost parameter $\hat{\gamma}$ and emissions reduction factor $\hat{\varrho}$.
	Accordingly, we can follow the proof of Proposition~\ref{prop:optimalquant} to establish the required result under the tax scheme.

	Under the assumption that $0<P/\lambda < 1$, for a given production $q$ and emissions reduction effort  $a$, the optimization problem for the certificate demand $\delta$ under the market scheme is given by
	\begin{align}
		 & \min_{\delta} \left\{ \delta P + \lambda \mathbb{E}\left[ \left( e^{-a}  q  X - \delta \right)^+ \right] \right\} =   P\ \min_{\delta} \left\{ \delta  + \tfrac{1}{P/\lambda} \mathbb{E}\left[ \left( e^{-a}  q  X - \delta \right)^+ \right] \right\} \notag
		\\
		 & \,  = P q e^{-a}  \min_{m} \left\{ m  + \tfrac{1}{P/\lambda} \mathbb{E}\left[ \left( X - m \right)^+ \right] \right\} \notag
		\\
		 & \,  = P q e^{-a} \left( \mathbb{V}{\rm a}\mathbb{R}_{P/\lambda}(X) + \tfrac{\lambda}{P} \mathbb{E}\left[ (X - \mathbb{V}{\rm a}\mathbb{R}_{P/\lambda}(X))^+ \right] \right)  = P q e^{-a} \, \mathbb{ES}_{P/\lambda} (X). \label{eq:ee_market_random}
	\end{align}
	This follows from the fact that expected shortfall is positive homogeneous and from the connections~\eqref{eq:oce_conn} between the value-at-risk and expected shortfall risk measures.
	Thus, the optimal certificate demand is uniquely given by $\bm{\delta}=  q e^{-a} \, \mathbb{V}{\rm a}\mathbb{R}_{P/\lambda}(X)$. 
	The uniqueness follows from the fact that the argmin $m^\ast =\mathbb{V}{\rm a}\mathbb{R}_{P/\lambda}(X)$ in the second line of the above equation is unique for a continuously distributed random variable $X$.
	By~\eqref{eq:ee_market_random}, the  company's expected wealth simplifies to
	\begin{equation}\label{eq:wealth_mar}
		\pi(q) - c(q, a) -  q e^{-a} P \, \mathbb{ES}_{P/\lambda} (X),
	\end{equation}
	which is equivalent to the wealth under a tax scheme, with tax rate set to $\tau = P \, \mathbb{ES}_{P/\lambda} (X)$.
	By analogy, we derive the expressions for the company's optimal production level $\bm{ q^{\mathrm{mar}} }$, emissions reduction rate $\bm{ a^{\mathrm{mar} } }$, carbon emissions (expected) $\mathrm{E}^{\mathrm{mar}}$, and net wealth (expected) $\mathcal{W}^{\mathrm{mar}}_{\mathrm{C}}$.
	Finally, the certificates demand simplifies to $\bm{\delta} = e^{-\bm{ a^{\mathrm{mar} } }} \bm{ q^{\mathrm{mar}} }\, \mathbb{V}{\rm a}\mathbb{R}_{P/\lambda}(X) = \mathbb{V}{\rm a}\mathbb{R}_{P/\lambda}(X)\mathrm{E}^{\mathrm{mar}}$.
\end{proof}

Proposition~\ref{prop:optimalquant_random} implies that, when emission is random, a company's demand for certificates is $\bm{\delta} =  \mathbb{V}{\rm a}\mathbb{R}_{P/\lambda}(X) \mathrm{E}^{\mathrm{mar}} $, which generally differs from its expected total emissions $\mathrm{E}^{\mathrm{mar}}$.
As a result, the company typically holds either fewer or more certificates than its average requirement, leading to an expected excess or shortfall of emissions---as opposed to the deterministic case, see Remark~\ref{rem:zero_loss}. 
When there are many companies indexed by $i$, Proposition~\ref{prop:optimalquant_random} applies to each firm, just as in the deterministic setting. 

Recall that $\bm{P_i}(S)$ denotes the effective price per certificate faced by company $i$ at a given spot price $S$.
Specifically, $\bm{P_i}(S)=S$ if the company purchases certificates directly at auction, while $\bm{P_i}(S)$ corresponds to the price that maximizes the intermediary's wealth when the company acquires certificates through an intermediary, as in the deterministic case.
If the regulator supplies $A$  certificates at auction, then under mild assumptions, there exists a unique (Radner \citep{radner1972}) equilibrium spot price $\bm{S}$ and firm-level effective price $\bm{P_i}(\bm{S})$ such that the clearing condition
$\sum \bm{\delta_i} =A$ is satisfied with $\bm{\delta_i} = \bm{\delta_i}(\bm{P_i}(S))$, see Propositions~\ref{prop:spot_random}--\ref{prop:unique_spot_random}.  

Since a company's expected emissions do not coincide with its certificate demand in the random case, a fair comparison of the tax and market schemes requires fixing a common expected emissions target $\hat{A}$, where $\hat{A}$ is a constant strictly larger than the number of auctioned certificates $A$.
The difference $\hat{A}-A$ thus represents the expected excess emissions tolerated by the regulator.
Accordingly, the emission constraint becomes
\begin{equation*}
	\mathrm{E}^{\mathrm{tax}} = \mathrm{E}^{\mathrm{mar}} =\hat{A}. 
\end{equation*}  

For the given tax rate $\tau$ and expected  emissions target $\hat{A}$, the regulator's wealth under the tax scheme is $\mathcal{W}^{\mathrm{tax}}_{\mathrm{R}} = \tau \hat{A}$.
Under the market scheme, for a given penalty rate $\lambda$,  total number of certificates  $A$, and effective certificate price $\bm{P_i} =\bm{P_i}(\bm{S})$ for each company $i$ with equilibrium spot price $\bm{S}$, the regulator's wealth is given by
\begin{align}
	\mathcal{W}^{\mathrm{mar}}_{\textrm{R}} & = \lambda \sum   \mathbb{E}\left[ \left( \bm{q^{\mathrm{mar}}_i} e^{-\bm{a^{\mathrm{mar}}_i}} X - \bm{ \delta_i} \right)^+ \right] + \bm{S}\ A  \notag
	\\
	                                     & = \sum \displaystyle \left[ \mathbb{ES}_{ \bm{P_i}/ \lambda }(X) - \mathbb{V}{\rm a}\mathbb{R}_{ \bm{P_i}/\lambda }(X) \right] \bm{P_i} \, \mathrm{E}^{\mathrm{mar}}_i + \bm{S}\  A,\label{revenue:regu:mar}
\end{align}
where the second equation follows from Eqn.~\eqref{eq:ee_market_random}.

The following theorem compares the economic outputs of the tax and market schemes in the case of random carbon emissions.
\begin{theorem}[\textbf{Tax vs.\ spot market scheme}]\label{thm:no_intermediaries_random}
	Assume that the tax rate $\tau$, penalty rate $\lambda$ and  number of certificates $A$ are  chosen such that targeted aggregate expected  emissions $\hat{A}>A$ are identical across all policy schemes.
	Then
	\begin{itemize}[fullwidth]
		\item \textbf{Tax scheme}:
		      The tax rate is given by $ \tau =  (\mathrm{E}^{\mathrm{bau}}-\hat{A})/\hat{\varrho}$, where $\hat{\varrho} = \sum \hat{\varrho}_i$,  and the expected wealth of company $i$ is
		      \begin{equation*}
			      \mathcal{W}^{\mathrm{tax}}_{\mathrm{C}, i} = \mathcal{W}^{\mathrm{bau}}_{\mathrm{C}, i} - \tau \left(\mathrm{E}^{\mathrm{bau}}_i - \frac{\tau}{2}\hat{\varrho}_i \right),
		      \end{equation*}
		      resulting into the aggregated expected  wealth  of the companies, 
			  and of the regulator given by
		      \begin{align}
			      \mathcal{W}^{\mathrm{tax}}_{\mathrm{C}} = \mathcal{W}^{\mathrm{bau}}_{\mathrm{C}} - \tau \left(\mathrm{E}^{\mathrm{bau}} - \frac{\tau}{2}\hat{ \varrho} \right), 
				  \quad 
			      \mathcal{W}^{\mathrm{tax}}_{\mathrm{R}}       = \tau \hat{A}.
		      \end{align}
		\item \textbf{Spot market scheme}:
		      Let $G$ denote the survival function of the random variable $X$. 
			  Under the assumptions of Proposition~\ref{prop:unique_spot_random}, the equilibrium spot price $\bm{S}$ is given by
		      \begin{equation*}
			      \bm{S}  = \lambda\ G\left(A/\hat{A} \right) = \dfrac{\tau}{ \mathbb{ES}_{G( A/\hat{A} )}(X) }.
		      \end{equation*}
		      The resulting expected wealth of company $i$ is  the same as under the tax scheme, i.e.
		      \begin{equation*}
			      \mathcal{W}^{\mathrm{mar}}_{\mathrm{C}, i} = \mathcal{W}^{\mathrm{tax}}_{\mathrm{C}, i}, 
		      \end{equation*}
		      with  the aggregated  wealth of the companies and wealth of the  regulator given by
		      \begin{align*}
			      \mathcal{W}^{\mathrm{mar}}_{\mathrm{C}}  = \mathcal{W}^{\mathrm{tax}}_{\mathrm{C}},
				  \quad 
			      \mathcal{W}^{\mathrm{mar}}_{\mathrm{R}}  = \mathcal{W}^{\mathrm{tax}}_{\mathrm{R}}.
		      \end{align*}
	\end{itemize}
\end{theorem}

\proof See Subsection~\ref{app:no_intermediaries_random}.

The above result shows that  the tax and spot market schemes remain equivalent even under emission uncertainty.

We now turn to compare the tax and purely intermediated market schemes.
Due to the significant computational challenges in  deriving equilibrium prices analytically under the purely intermediated market scheme, we focus on the symmetric setup of Example~\ref{eg:symetric}, but with random emissions here.
For the general case, see the numerical results in Section~\ref{sec:numeric_random}.

\begin{proposition}[\textbf{Tax vs.\ purely intermediated market scheme}]\label{prop:with_intermediaries_random}
	Assume that the tax rate $\tau$, penalty rate $\lambda$ and  number of certificates $A$ are  chosen such that the expected targeted aggregate emissions $\hat{A}$ are identical across all policy schemes.
	If all companies are identical, i.e. all the company-specific parameters share the same values, then:
	\begin{itemize}[fullwidth]
		\item  The equilibrium spot price $\bm{S}$ and the intermediated price $\bm{P_i}(\bm{S})$ are given by
		      \begin{align*}
			       & \bm{S}  =   \lambda\ G\left( A / \hat{A}  \right)-  \tfrac{\lambda A }{\lambda N\hat{\varrho}_1 \left( \tfrac{A}{\hat{A} } \right)^2  -  \tfrac{\hat{A}}{G'\left(A/ \hat{A}\right)}}, \quad \text{with} \quad \hat{\varrho}_i = \hat{\varrho}_1,
			      \\[1.5ex]
			       & \bm{P_i}(\bm{S})    = \tfrac{\tau}{ \mathbb{ES}_{G( A/\hat{A})}(X) }= \lambda\ G\left(A/ \hat{A}\right).
		      \end{align*}

		\item The expected wealth of the financial intermediary serving company $i$ and  of the company $i$  are  given by 
		      \begin{equation*}
			     \mathcal{W}^{\mathrm{mar}}_{\mathrm{F}, i}    =   \tfrac{\lambda \ A^2}{\lambda \hat{\varrho_1} \left( \tfrac{N\ A}{\hat{A}} \right)^2  - \tfrac{N\ \hat{A}}{G'\left(A/\hat{A} \right)}},
				 \quad 
				 	 		\mathcal{W}^{\mathrm{mar}}_{\mathrm{C}, i}  = \mathcal{W}^{\mathrm{tax}}_{\mathrm{C}, i} = \mathcal{W}^{\mathrm{tax}}_{\mathrm{C}, 1}.
		      \end{equation*}
		      This leads to  a total wealth for the financial intermediaries, the companies  and the regulator  given by
		      \begin{equation*}      
			      \mathcal{W}^{\mathrm{mar}}_{\mathrm{C}}  = \mathcal{W}^{\mathrm{tax}}_{\mathrm{C}},
					\quad 
						\mathcal{W}^{\mathrm{mar}}_{\mathrm{F}}  = \mathcal{W}^{\mathrm{mar}}_{\mathrm{F}, 1} ,
			        		 \quad  
			     	 \mathcal{W}^{\mathrm{mar}}_{\mathrm{R}}       =  \mathcal{W}^{\mathrm{tax}}_\mathrm{R}   - \mathcal{W}^{\mathrm{mar}}_{\mathrm{F}}.
		      \end{equation*}
	\end{itemize}
\end{proposition}

\proof See Subsection~\ref{app:with_intermediaries_random}.

Proposition~\ref{prop:with_intermediaries_random} shows that, in the symmetric case, both schemes generate the same wealth for companies, whereas the tax scheme yields a greater wealth for the regulator. Consequently, the GDP of the economy under the market scheme is strictly lower than under the tax scheme.

\subsection{Numerical Illustration}\label{sec:numeric_random}

For the numerical computations, we assume that the emissions $X \sim \mathrm{LogNormal}(-\tfrac{1}{2}, 1)$.
The corresponding survival function $G$ and its derivative are
\begin{equation*}
	G(t)  = \Phi\left( -\ln(t) - \tfrac{1}{2} \right),
	\qquad
	G'(t) = -\frac{\Phi'\left( \ln t + \tfrac{1}{2} \right)}{t}, \quad  t>0
\end{equation*}
where $\Phi$ is the cumulative distribution function of a standard normal distribution.
In this case, the second moment of $X$ is $\sigma^2 = \mathbb{E}[X^2] = e$.
Moreover, for $0 < \epsilon \leq 1$, the value-at-risk and expected shortfall are given by
\begin{equation*}
	\mathbb{V}{\rm a}\mathbb{R}_\epsilon(X)  = \exp\left( -\Phi^{-1}(\epsilon) - \tfrac{1}{2} \right),
	\qquad
	\mathbb{ES}_\epsilon(X)   = \frac{\Phi\left( \Phi^{-1}(\epsilon) + 1 \right)}{\epsilon}.
\end{equation*}

For our computations, we set the expected emissions target to $\hat{A} = 1.05 A$, while all other parameter values are as specified in Table~\ref{tab:parameters}.

As shown in Table~\ref{tab:value_rand}, and Figures~\ref{fig:gdp_tax_vs_fi_rand} and \ref{fig:emission_tax_vs_fi_rand},  with random emissions, outcomes remain broadly comparable to those observed in the deterministic case.
Under the market scheme, companies primarily achieve emission reductions through production cuts, as evidenced by their lower investment in green technologies compared with the tax scheme (€3 billion under the market scheme versus €14 billion under the tax scheme),  see Table~\ref{tab:value_rand}.
This indicates that the tax scheme provides stronger incentives for companies to reduce emissions through investment in green technologies under conditions of uncertain emissions.

\begin{table}[htp!]
    \centering
    \begin{threeparttable}
        \caption{GDP and emission contributions of each actor in the business-as-usual case and under  the tax and market schemes with random emissions.}
        \label{tab:value_rand}
        \begin{tabular}{@{}lrrrcrrr@{}}
            \toprule
            & \multicolumn{3}{c}{ \makecell{ GDP contributions\\ (billion  €) } } & & \multicolumn{3}{c}{ \makecell{Emission contributions \\ (million tons)} } \\
            \cmidrule{2-4} \cmidrule{6-8}
            Group & $\mathrm{bau}$ & $\mathrm{tax}$ & $\mathrm{mar}$ & & $\mathrm{bau}$ & $\mathrm{tax}$ & $\mathrm{mar}$ \\
            \midrule
            Brown & 300 & 57 & 61 & & 1,051 & 365 & 392 \\
            \quad B1 & 100 & 10 & 14 & & 420 & 105 & 135 \\
            \quad B2 & 100 & 20 & 19 & & 330 & 139 & 136 \\
            \quad B3 & 100 & 28 & 28 & & 300 & 121 & 120 \\
            Green & 300 & 112 & 101 & & 721 & 379 & 352 \\
            \quad G1 & 100 & 31 & 28 & & 270 & 133 & 124 \\
            \quad G2 & 100 & 35 & 31 & & 240 & 137 & 126 \\
            \quad G3 & 100 & 45 & 42 & & 210 & 109 & 102 \\
            Regulator & --- & 255 & 186 & & --- & --- & --- \\
            Tech providers & --- & 14 & 3 & & --- & --- & --- \\
            Financial intermediaries & --- & --- & 75 & & --- & ---& --- \\
            \midrule
            \textbf{Total} & 600 & 438 & 426 & & 1,772 & 744 & 744 \\
            \bottomrule
            \end{tabular}
        \end{threeparttable}
    \end{table}
\begin{figure}[htp!]
	\centering
	\includegraphics[width=1.0\textwidth]{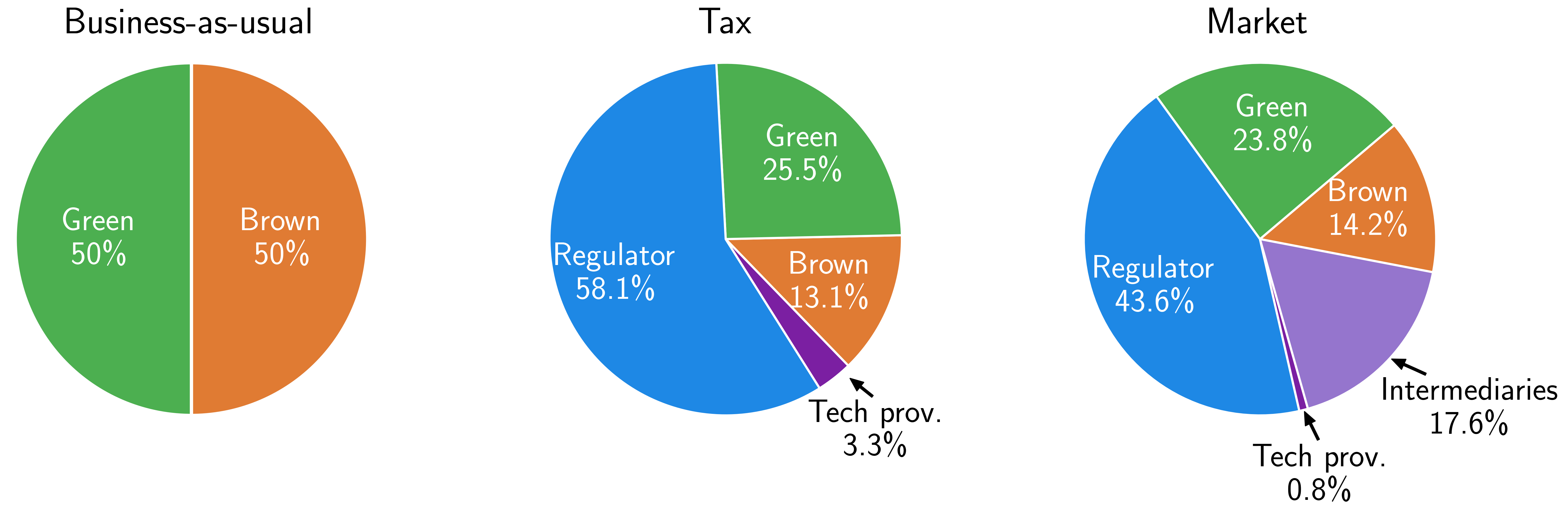}
	\caption{The GDP shares of the different  actors in the business-as-usual case and under the  tax and purely intermediated market schemes with random emissions.}
	\label{fig:gdp_tax_vs_fi_rand}
\end{figure}

\begin{figure}[htp!]
	\centering
	\includegraphics[width=1.0\textwidth]{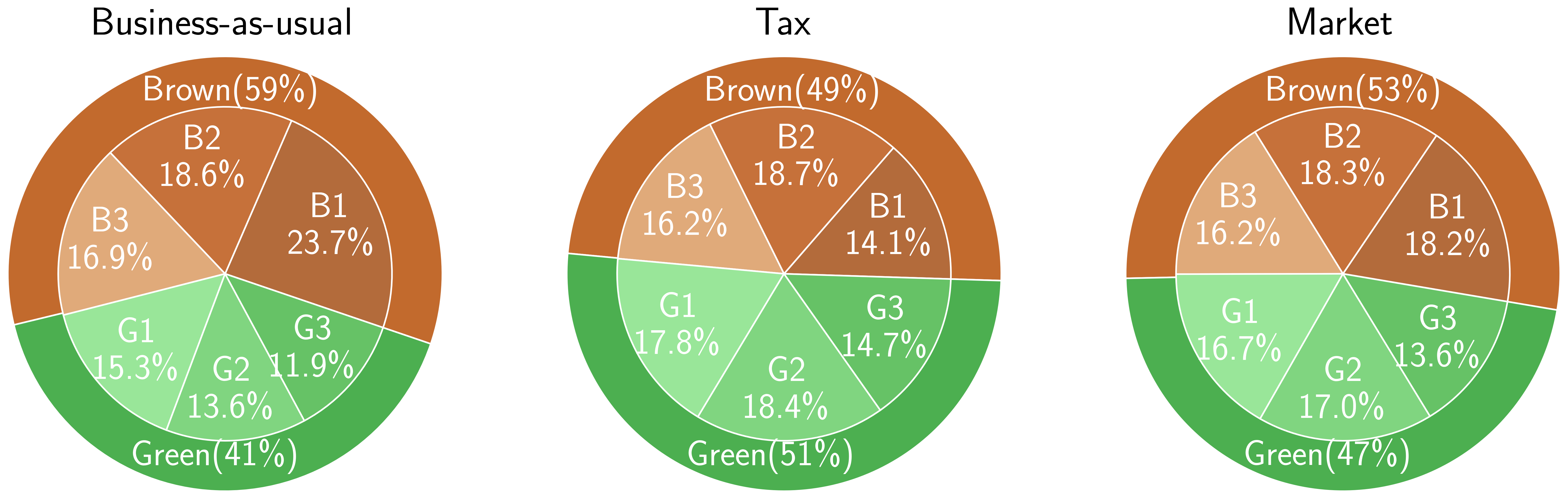}
	\caption{The emission shares of companies in the business-as-usual case and under the tax and  purely intermediated market schemes with random emissions.}
	\label{fig:emission_tax_vs_fi_rand}
\end{figure}

\section{Conclusion}\label{sec:conclusion}

This paper provides an analytical framework for comparing carbon taxation (tax scheme) with emissions trading systems (market scheme), with a specific focus on the role of financial intermediaries within the latter. By developing a one-period equilibrium model, we assess the economic and environmental performance of these schemes under the constraint that they achieve identical aggregate emissions reductions.

Our findings confirm that under perfect competition, where companies have direct and frictionless access to the carbon certificates' central auction, the two schemes are equivalent.
This equivalence holds true regardless of whether emissions are deterministic or stochastic.

However, this theoretical parity is broken by the introduction of financial intermediaries.
When intermediaries are the primary channel for companies to acquire emission certificates, they influence price formation and capture a portion of the revenue stream.
This leads to key divergences from the tax scheme:
\begin{itemize}
	\item \textbf{Regulatory Wealth}: The presence of intermediaries strictly reduces the wealth collected by the regulatory authority compared to an equivalent tax scheme, as they appropriate a portion of the value that would otherwise accrue to the regulator through auction proceeds.
	      This aligns with recent empirical observations of diverging wealth trends between ETS and carbon taxes.
	\item \textbf{Aggregate Economic Output}: The combined wealth of companies, financial intermediaries, and the regulator are lower under the intermediated market scheme than under the tax scheme, indicating a net loss of economic value in the real economy due to market friction.
	\item \textbf{Distributional Effects}: While aggregate company wealth are strictly lower under the market scheme, individual company profitability relative to the tax scheme depends on specific cost and production factors.
\end{itemize}

Our numerical analysis also suggests that the two schemes incentivize abatement differently.
Under the intermediated market scheme, high-intensity companies maintain relatively high emissions, whereas the tax scheme more effectively incentivizes these companies to drive the majority of emissions reductions.
Additionally, under emission uncertainty, the tax scheme provides stronger incentives for green technology investment, while the market scheme primarily encourages production cuts.

These results highlight that the choice of carbon pricing instrument is not solely a matter of economic efficiency but also of institutional design and its resulting distributional consequences.
The study underscores the importance of accounting for intermediary behavior in the design of carbon markets to inform policy on regulation and the distribution of economic burdens and benefits.
Future research could extend this framework to multi-period models to capture temporal dynamics, including the allocation of certificates and the role of derivative markets.

\appendix

\section{Some Supplementary results on Section~\ref{sec:model}} \label{app:model}

Throughout this section, we consider a \emph{hybrid market scheme}, in which some companies acquire certificates through financial intermediaries, while others purchase them directly via market auction.
To distinguish between these two groups, we use the indices $b$ and $c$ to denote, respectively, companies that engage in bilateral trades with intermediaries and those that participate directly in the market.

\begin{proposition}\label{prop:given_S}
	Let $S \in (0, \lambda]$ be the given spot price of a unit certificate on the market.
	It holds that
	\begin{itemize}
		\item \textbf{Direct participation to the auction}:
		      When a company $c$ purchases certificates directly at a given spot price $S$, its optimal demand $\bm{\delta_c}$ satisfies
		      \begin{equation*}
			     \bm{\delta_c}(\bm{P_c}(S)) =  \bm{\delta_c}(S) = \mathrm{E}^{\mathrm{bau}}_c - S \varrho_c.
		      \end{equation*}
		\item \textbf{Purchase through a financial intermediary}: When a company $b$ purchases certificates through a financial intermediary, the intermediated price $\bm{P_b}(S)$ for a given spot price $S$ is
		      \begin{equation*}
			      \bm{P_b}(S)   = \frac{1}{2} \left(S + \frac{\mathrm{E}^{\mathrm{bau}}_b}{\varrho_b} \right)\wedge \lambda
		      \end{equation*}
		      with  resulting  certificates demand $\bm{\delta_b}$ such that	      
              \begin{equation*}
			       \bm{\delta_b}(\bm{P_b}(S)) = \tfrac{\left( \mathrm{E}^{\mathrm{bau}}_b- S \varrho_b \right)}{2} \vee \left( \mathrm{E}^{\mathrm{bau}}_b- \lambda \varrho_b \right) .
		      \end{equation*}
	\end{itemize}
\end{proposition}

\begin{proof}
	Let $S \in (0, \lambda]$ be a given spot price.
	For companies that participate directly in the certificate auction (indexed by $c$), their optimal demand follows immediately from Proposition~\ref{prop:optimalquant}.

	For companies purchasing certificates through financial intermediaries (indexed by $b$), substituting the optimal certificate demand $\bm{\delta_b}(P)$ from Proposition~\ref{prop:optimalquant} into the wealth function of the intermediary serving company $b$ yields
	\begin{equation*}
		\bm{\delta_b}(P) (P -S) = \left( \mathrm{E}^{\mathrm{bau}}_b + \varrho_b S \right) P  - \varrho_b  P^2  - S\ \mathrm{E}^{\mathrm{bau}}_b,
	\end{equation*}
	which is strictly concave  in $P$.
	The unique maximizer of the map $P \mapsto \bm{\delta_b}(P)(P - S)$ is characterized by the first-order condition and is given by
	\begin{equation*}
		P^\dagger =  \frac{1}{2} \left(S + \frac{\mathrm{E}^{\mathrm{bau}}_b}{\varrho_b} \right).
	\end{equation*}
	The assumption $\lambda \varrho_b < \mathrm{E}^{\mathrm{bau}}_b$ ensures that $P^\dagger \geq S$.
	Since $P^\dagger$ is the unique critical point (and global maximizer) of the strictly concave function $P \mapsto \bm{\delta_b}(P)(P - S)$, the optimal intermediated price $\bm{P_b}(S)$---maximizing the intermediary's wealth over the admissible interval $[S, \lambda]$---is given by
	\begin{equation*}
		\bm{P_b}(S) =   P^\dagger \wedge \lambda.
	\end{equation*}
	This yields the required characterization of the intermediated price.
	By Proposition~\ref{prop:optimalquant}, the corresponding optimal certificate demand is
	\begin{align*}
		\bm{\delta_b}(\bm{P_b}(S))& = \mathrm{E}^{\mathrm{bau}}_b - \bm{P_b}(S)\varrho_b  =  \tfrac{\left(\mathrm{E}^{\mathrm{bau}}_b -S\varrho_b \right)}{2}\vee \left(\mathrm{E}^{\mathrm{bau}}_b -\lambda \varrho_b\right),
	\end{align*}
	which concludes the proof of the proposition.
\end{proof}

In view of Proposition~\ref{prop:given_S}, for a given spot price $S$, the clearing condition is
\begin{equation}\label{eq:mar_cc}
	A = \text{\emph{clear}} (S) \colon = \underbrace{\sum \left( \mathrm{E}^{\mathrm{bau}}_c - S \varrho_c \right)}_{\parbox{3cm}{\footnotesize \centering demand from direct market participants} } + \; \underbrace{\sum \tfrac{\left( \mathrm{E}^{\mathrm{bau}} _b- S \varrho_b \right)}{2} \vee \left( \mathrm{E}^{\mathrm{bau}} _b- \lambda \varrho_b \right)}_{\parbox{4cm}{\footnotesize \centering  demand from intermediated purchases}}.
\end{equation}
Hence, the equilibrium spot price $\bm{S}$ lies in the interval $(0, \lambda]$ if and only if the certificate supply $A$ satisfies
\begin{equation}\label{eq:A}
	\underbrace{\mathrm{E}^{\mathrm{bau}} - \lambda \varrho}_{=\text{\emph{clear}} (\lambda)} \; \leq  A
	< \; \underbrace{\sum \mathrm{E}^{\mathrm{bau}}_c
	+ \sum \left[ \tfrac{ \mathrm{E}^{\mathrm{bau}}_b }{2} \vee \left( \mathrm{E}^{\mathrm{bau}}_b - \lambda \varrho_b \right) \right]}_{= \displaystyle\lim_{S\searrow 0} \text{\emph{clear}} (S)},
\end{equation}
where $\mathrm{E}^{\mathrm{bau}} = \sum \mathrm{E}^{\mathrm{bau}}_b + \sum \mathrm{E}^{\mathrm{bau}}_c $ and $\varrho = \sum \varrho_b + \sum \varrho_c $.

\begin{remark}
	Under the spot market scheme, the feasibility condition~\eqref{eq:A} simplifies to
	\begin{equation*}
		\mathrm{E}^{\mathrm{bau}} - \lambda \varrho \leq A  < \mathrm{E}^{\mathrm{bau}}.
	\end{equation*}
	Under the purely intermediated market scheme, \eqref{eq:A} reduces to
	\begin{equation*}
		\mathrm{E}^{\mathrm{bau}} - \lambda \varrho \leq A  <  \sum \left[ \tfrac{ \mathrm{E}^{\mathrm{bau}}_b }{2} \vee \left( \mathrm{E}^{\mathrm{bau}}_b - \lambda \varrho_b \right) \right].
	\end{equation*}
\end{remark}

We make the following observations regarding the bounds  ~\eqref{eq:A} on $A$:
\begin{itemize}[fullwidth]
	\item The inequality $\mathrm{E}^{\mathrm{bau}}_b - \lambda \varrho_b \geq \mathrm{E}^{\mathrm{bau}}_b/2$ (whenever satisfied) means that, at the maximum spot price $S = \lambda$, the certificate demand of company $b$ remains at least half of its baseline emissions.

	\item
	      Under a purely intermediated market scheme, if $\mathrm{E}^{\mathrm{bau}}_b - \lambda \varrho_b \geq \mathrm{E}^{\mathrm{bau}}_b/2$ holds for all companies $b$, then the admissible set for the number of certificates $A$ becomes empty.
	      Consequently, no market-clearing price $S \in (0, \lambda]$ exists in such configuration.

	\item For a given spot price $S \in (0, \lambda]$:
	      \begin{equation*}
		      \bm{P_b}(S)   = \frac{1}{2} \left(S + \frac{\mathrm{E}^{\mathrm{bau}}_b}{\varrho_b} \right)\wedge \lambda = \frac{1}{2} \left(S + \frac{\mathrm{E}^{\mathrm{bau}}_b}{\varrho_b} \right) \iff 0 < S \leq 2\lambda - \frac{\mathrm{E}^{\mathrm{bau}}_b}{\varrho_b}.
	      \end{equation*}
	      Hence, this inequality holds for all $b$, if and only if
	      \begin{equation*}
		      0 < S \leq 2\lambda - \max_b \left( \frac{\mathrm{E}^{\mathrm{bau}}_b}{\varrho_b} \right).
	      \end{equation*}
	      Using the market-clearing condition~\eqref{eq:mar_cc}, this requirement translates into the following condition on $A$:
	      \begin{equation*}
		      \text{\emph{clear}} (2\lambda - \max_b (\mathrm{E}^{\mathrm{bau}}_b / \varrho_b)) \leq A < \lim_{S \searrow 0} \text{\emph{clear}} (S),
	      \end{equation*}
	      i.e.
	      \begin{multline}\label{eq:adequate}
		      \left[ \max_b\left(\frac{\mathrm{E}^{\mathrm{bau}}_b}{\varrho_b}\right) - 2\lambda \right] \left( \sum \varrho_c  + \frac{1}{2}\sum \varrho_b \right) + \sum \mathrm{E}^{\mathrm{bau}}_c  + \frac{1}{2} \sum \mathrm{E}^{\mathrm{bau}}_b
		       \\
		      \leq A < \sum \mathrm{E}^{\mathrm{bau}}_c + \sum \left[ \tfrac{\mathrm{E}^{\mathrm{bau}}_b}{2} \vee \left( \mathrm{E}^{\mathrm{bau}}_b - \lambda \varrho_b \right) \right].
	      \end{multline}

	\item Under the condition~\eqref{eq:adequate}, by Proposition~\ref{prop:given_S}, the intermediated price and the resulting certificate demand of each company $b$ are reduced to
	      \begin{align}\label{eq:simplified}
		      \bm{P_b}(S)           & = \frac{1}{2} \left(S + \frac{\mathrm{E}^{\mathrm{bau}}_b}{\varrho_b} \right),
		                       &
		      \bm{\delta_b}(\bm{P_b}(S)) & = \frac{1}{2} \left( \mathrm{E}^{\mathrm{bau}}_b - S \varrho_b \right),
	      \end{align}
	      ensuring an interior (non-corner) solution for both the intermediated price and certificates demand of company $b$.
\end{itemize}

\begin{remark}
	Under the purely intermediated market scheme, this feasibility condition~\eqref{eq:adequate}  reduces to
	\begin{equation*}
		\frac{1}{2}\left[\mathrm{E}^{\mathrm{bau}} + \varrho\left( \max_b\left(\frac{\mathrm{E}^{\mathrm{bau}}_b}{\varrho_b} \right) -2\lambda \right)  \right] \leq A  <  \sum \left[ \tfrac{\mathrm{E}^{\mathrm{bau}}_b }{2} \vee \left( \mathrm{E}^{\mathrm{bau}}_b- \lambda \varrho_b \right) \right],
	\end{equation*}
	where the lower bound involves  the most emission-intensive company through the term $\max_b(\mathrm{E}^{\mathrm{bau}}_b/\varrho_b)$. This captures how intermediated trading shapes the admissible interval for market clearing.
\end{remark}

\begin{proposition}[\textbf{Hybrid market scheme}]\label{prop:equ_spot}
	Under the feasibility condition~\eqref{eq:adequate} on the number of certificates $A$, there exists a unique clearing spot price $\bm{S} \in (0, \lambda]$ given by
	\begin{align*}
		\bm{S} & =  \dfrac{ 2 (\mathrm{E}^{\mathrm{bau}}  - A) - \sum\mathrm{E}^{\mathrm{bau}}_b }{ 2\varrho- \sum \varrho_b },
	\end{align*}
	with the resulting carbon price and certificates demand:
	\begin{align*}
		& \bm{P_c}(\bm{S})  = \bm{S}, 
		\qquad 
			\bm{\delta_c}(\bm{P_c}(\bm{S}))  = \bm{\delta}_c(\bm{S})=    \mathrm{E}^{\mathrm{bau}}_c - \bm{S} \varrho_c, \; \forall c, \\
		& \bm{P_b}(\bm{S}) = \dfrac{1}{2} \left(\bm{S} + \dfrac{\mathrm{E}^{\mathrm{bau}}_b}{\varrho_b} \right),  
		\qquad
		\bm{\delta_b}(\bm{P_b}(\bm{S}))  = \frac{1}{2}  \left( \mathrm{E}^{\mathrm{bau}}_b -  \bm{S}\varrho_b \right), \; \forall b .
	\end{align*}
\end{proposition}

\begin{proof}
	If the number of certificates $A$ satisfies the feasibility condition~\eqref{eq:adequate}, then the clearing function $\text{\emph{clear}} (S)$ is bijective on the interval $(0, \lambda]$.
	As a result, there exists a unique equilibrium spot price $\bm{S}$ that clears the market.
	Under this condition, and using~\eqref{eq:simplified}, we obtain the noted expressions for the intermediated price $\bm{P_b}(\bm{S})$ and the corresponding certificate demand $\bm{\delta_b}(\bm{P_i}(\bm{S}))$ for companies purchasing certificates through intermediaries.
	The result for companies purchasing certificates directly on the market follows immediately from Proposition~\ref{prop:given_S}.
	Under this setting, the market clearing condition becomes
	\begin{equation*}
		A = \sum \left( \mathrm{E}^{\mathrm{bau}}_c - S \varrho_c \right) + \frac{1}{2} \sum \left( \mathrm{E}^{\mathrm{bau}}_b - S\varrho_b \right).
	\end{equation*}
	Solving this equation for the spot price $S$ yields the desired result.
\end{proof}

\section{Some Supplementary results on Section~\ref{sec:compare}}\label{app:hybridcase}

\begin{theorem}[\textbf{Hybrid market scheme}]\label{thm:comparison_hybrid} Assume that the tax rate $\tau$, penalty rate $\lambda$, and number of certificates $A$ are chosen such that  aggregate emissions are identical across all policy schemes.
	Then
	\begin{itemize}[fullwidth]
		\item \textbf{Tax scheme}:
		      The tax rate is given by $ \tau =  (\mathrm{E}^{\mathrm{bau}}-A)/\varrho$ and the  wealth of company $i$ is
		      \begin{equation*}
			      \mathcal{W}^{\mathrm{tax}}_{\mathrm{C}, i} = \mathcal{W}^{\mathrm{bau}}_{\mathrm{C}, i} - \tau \left(\mathrm{E}^{\mathrm{bau}}_i - \frac{\tau}{2}\varrho_i \right),
		      \end{equation*}
		      where the index $i$ is arbitrary ($i=b, c$).
		      This leads to the aggregate wealth of the companies and  a wealth for the regulator given by
		      \begin{equation*}
			      \mathcal{W}^{\mathrm{tax}}_{\mathrm{C}}  = \mathcal{W}^{\mathrm{bau}}_{\mathrm{C}} - \tau \left(\mathrm{E}^{\mathrm{bau}} - \frac{\tau}{2}\varrho \right),
			      \qquad
			      \mathcal{W}^{\mathrm{tax}}_{\mathrm{R}}       = \tau A.
		      \end{equation*}

		\item \textbf{Market scheme}: Under the feasibility condition~\eqref{eq:adequate} on  $A$:
		      \begin{itemize}[label=$\blacktriangleright$]
			      \item For each company $c$  purchasing certificates directly on the market, the certificate price it faces is $\bm{P_c} (\bm{S}) = \bm{S}= \frac{2\varrho \tau- \sum \mathrm{E}^{\mathrm{bau}}_b  }{2 \varrho - \sum \varrho_b}$,
			            with resulting  wealth of
			            \begin{equation*}
				            \mathcal{W}^{\mathrm{mar}}_{\mathrm{C}, c}
				            = \mathcal{W}^{\mathrm{tax}}_{\mathrm{C}, c} +   \frac{ \left(\mathrm{E}_c^{\mathrm{tax}} + \mathrm{E}_c^{\mathrm{mar}} \right)}{2 \varrho} \sum\mathrm{E}_b^{\mathrm{mar}}.
			            \end{equation*}

			      \item For each company $b$  purchasing certificates through intermediary, the intermediated price it faces is $\bm{P_b}(\bm{S}) =  \frac{\tau}{2}+ \frac{1}{2} \left( \frac{\mathrm{E}^{\mathrm{bau}}_b }{ \varrho_b}  - \frac{1}{ \varrho} \sum \mathrm{E}_\beta^{\mathrm{mar}} \right) $,   with resulting wealth of the intermediary serving company $b$ and  wealth of  company $b$ given  by
			            \begin{align*}
				            \mathcal{W}^{\mathrm{mar}}_{\mathrm{F}, b} & = \frac{(\mathrm{E}_b^{\mathrm{mar}})^2}{\varrho_b},                                                         \\
				            \mathcal{W}^{\mathrm{mar}}_{\mathrm{C}, b} & =  \mathcal{W}^{\mathrm{tax}}_{\mathrm{C}, b} +   \frac{ \left(\mathrm{E}_b^{\mathrm{tax}} + \mathrm{E}_b^{\mathrm{mar}} \right)}{2\varrho} \sum \mathrm{E}_\beta^{\mathrm{mar}}  -\frac{3}{2} \mathcal{W}^{\mathrm{mar}}_{{\mathrm{F}}, b} + \frac{\mathrm{E}_b^{\mathrm{mar}} }{2 \varrho}  \sum \mathrm{E}_\beta^{\mathrm{mar}},
			            \end{align*}
			            where $\beta$ in the sum runs over all the companies that purchase certificates through financial intermediaries.
			      \item \textbf{Aggregated wealth}:
			            The aggregated wealth of the financial intermediaries, of the companies,  and the wealth of the  regulator, are given by
			            \begin{align*}
				            \mathcal{W}^{\mathrm{mar}}_{\mathrm{F}} & = \sum \frac{(\mathrm{E}_b^{\mathrm{mar}})^2}{\varrho_b},                                   \\
				            \mathcal{W}^{\mathrm{mar}}_{\mathrm{C}} & = \mathcal{W}^{\mathrm{tax}}_{\mathrm{C}} - \frac{3}{2}  \mathcal{W}^{\mathrm{mar}}_{\mathrm{F}} +   \frac{A}{\varrho}  \sum \mathrm{E}_b^{\mathrm{mar}}  + \frac{\left(\sum \mathrm{E}_b^{\mathrm{mar}} \right)^2}{2 \varrho}, \\
				            \mathcal{W}^{\mathrm{mar}}_{\mathrm{R}}      & = \mathcal{W}^{\mathrm{tax}}_{\mathrm{R}} - \frac{A}{\varrho} \sum  \mathrm{E}_b^{\mathrm{mar}}.
			            \end{align*}
		      \end{itemize}
	\end{itemize}
\end{theorem}

\begin{proof}
	Under the tax scheme, by Proposition~\ref{prop:optimalquant} and the assumption that aggregate emissions equal the certificate cap, i.e.  $\mathrm{E}^{\mathrm{tax}} = A$ (with $i = b, c$), we obtain
	\begin{equation*}
		\mathrm{E}^{\mathrm{tax}} = \sum \left( \mathrm{E}^{\mathrm{bau}}_i - \tau \varrho_i \right) = \mathrm{E}^{\mathrm{bau}} - \tau \varrho = A,
	\end{equation*}
	which implies that $\tau =(\mathrm{E}^{\mathrm{bau}}-A )/\varrho$.
	For given tax rate $\tau$ and certificates price $P_i$, by Proposition~\ref{prop:optimalquant}, the wealth of the generic company $i$ (with $i=b,c$) under the tax and any market schemes satisfy
	\begin{align*}
		\mathcal{W}^{\mathrm{tax}}_{\mathrm{C}, i} & = \mathcal{W}^{\mathrm{bau}}_{ {\mathrm{C}}, i}- \tau \left(\mathrm{E}^{\mathrm{bau}}_i - \frac{\tau}{2} \varrho_i \right),
		                                        &
		\mathcal{W}^{\mathrm{mar}}_{\mathrm{C}, i} & = \mathcal{W}^{\mathrm{bau}}_{ {\mathrm{C}}, i}- P_i \left(\mathrm{E}^{\mathrm{bau}}_i - \frac{P_i}{2} \varrho_i \right).
	\end{align*}
	It follows that
	\begin{align*}
		\mathcal{W}^{\mathrm{mar}}_{\mathrm{C}, i} & =  \mathcal{W}^{\mathrm{tax}}_{\mathrm{C}, i}  +  \tau \left(\mathrm{E}^{\mathrm{bau}}_i- \frac{\tau}{2} \varrho_i \right)- P_i \left(\mathrm{E}^{\mathrm{bau}}_i - \frac{P_i}{2} \varrho_i \right) \\
		                                        & = \mathcal{W}^{\mathrm{tax}}_{\mathrm{C}, i}  + ( \tau -P_i)\mathrm{E}^{\mathrm{bau}}_i - \frac{1}{2}( \tau^2 -P^2_i)\varrho_i                                                                           \\
		                                        & = \mathcal{W}^{\mathrm{tax}}_{\mathrm{C}, i}  + ( \tau -P_i)\left( \mathrm{E}^{\mathrm{bau}}_i - \frac{1}{2}(\tau + P_i)\varrho_i \right).
	\end{align*}
	This, together with the relations $\mathrm{E}^{\mathrm{tax}}_i = \mathrm{E}^{\mathrm{bau}}_i - \tau \varrho_i$ and $\mathrm{E}^{\mathrm{mar}}_i = \mathrm{E}^{\mathrm{bau}}_i - P_i \varrho_i$ from Proposition~\ref{prop:optimalquant}, yields
	\begin{equation}\label{eq:taxmar}
		\mathcal{W}^{\mathrm{mar}}_{\mathrm{C}, i}  =  \mathcal{W}^{\mathrm{tax}}_{\mathrm{C}, i}  + ( \tau -P_i)  \frac{ \left(\mathrm{E}^{\mathrm{tax}}_i + \mathrm{E}^{\mathrm{mar}}_i\right)}{2}.
	\end{equation}

	We now turn to the use of the specific indices $b$ and $c$ to distinguish companies under the market scheme.
	Substituting $\mathrm{E}^{\mathrm{bau}} - A = \tau \varrho$ into the expression of the spot price $\bm{S}$ given in Proposition~\ref{prop:equ_spot} yield
	\begin{align}
		\bm{P_c}  & = \bm{S} = \frac{2\varrho \tau- \sum \mathrm{E}^{\mathrm{bau}}_b  }{2 \varrho - \sum \varrho_b}, \notag                                                                  \\
		\tau & = \bm{S} + \frac{1}{\varrho}\sum \frac{1}{2}(\mathrm{E}^{\mathrm{bau}}_b - \bm{S} \varrho_b)  = \bm{S} + \frac{1}{\varrho}\sum \mathrm{E}^{\mathrm{mar}}_b,  \label{eq:S_tau}
	\end{align}
	where the second identity for $\tau$ follows from Proposition~\ref{prop:equ_spot}, using the identity  $\bm{\delta_b} = \mathrm{E}^{\mathrm{mar}}_b$ from Proposition~\ref{prop:optimalquant}.
	It implies that $\tau - \bm{P_c} = (1/\varrho)\sum \mathrm{E}^{\mathrm{mar}}_b$.
	This together with~\eqref{eq:taxmar} with $i=c$ yields
	\begin{equation*}
		\mathcal{W}^{\mathrm{mar}}_{\mathrm{C}, c} =  \mathcal{W}^{\mathrm{tax}}_{\mathrm{C}, c} +   \dfrac{ \left(\mathrm{E}_c^{\mathrm{tax}} + \mathrm{E}_c^{\mathrm{mar}} \right)}{2 \varrho} \sum  \mathrm{E}_b^{\mathrm{mar}}  .
	\end{equation*}
	From~\eqref{eq:S_tau}, we obtain $\bm{S} = \tau-  (1/\varrho)\sum \mathrm{E}^{\mathrm{mar}}_b$.
	Substituting this value for $\bm{S}$ into the expression of $\bm{P_b}$ from Proposition~\ref{prop:equ_spot} yields
	\begin{equation*}
		\bm{P_b} = \frac{\tau}{2}+ \frac{1}{2} \left( \frac{\mathrm{E}^{\mathrm{bau}}_b }{ \varrho_b}  - \sum \frac{\mathrm{E}_\beta^{\mathrm{mar}} }{ \varrho} \right),
	\end{equation*}
	where $\beta$ in the sum runs over all the companies that purchase certificates through financial intermediaries.
	Hence,
	\begin{equation*}
		\tau  = 2 \bm{P_b}-\frac{\mathrm{E}^{\mathrm{bau}}_b }{ \varrho_b}  + \frac{1}{ \varrho}\sum \mathrm{E}_\beta^{\mathrm{mar}},
		\qquad
		\bm{S}    = 2\bm{P_b} - \frac{\mathrm{E}^{\mathrm{bau}}_b }{ \varrho_b}.
	\end{equation*}
	This,  via $\mathrm{E}^{\mathrm{mar}}_b = \mathrm{E}^{\mathrm{bau}}_b -\bm{P_b} \varrho_b$ from Proposition~\ref{prop:optimalquant}  implies that
	\begin{align}\label{eq:diff}
		\bm{P_b} - \bm{S}   =\frac{\mathrm{E}^{\mathrm{bau}}_b }{ \varrho_b} -\bm{P_b} =  \frac{\mathrm{E}_b^{\mathrm{mar}} }{\varrho_b},
		\quad
		\tau - \bm{P_b}  = - \frac{\mathrm{E}_b^{\mathrm{mar}} }{\varrho_b} + \frac{1}{ \varrho}\sum \mathrm{E}_\beta^{\mathrm{mar}} .
	\end{align}
	Using the identity $\bm{\delta_b} = \mathrm{E}_b^{\mathrm{mar}}$ and substituting the expression for $\bm{P_b} - \bm{S}$ of~\eqref{eq:diff} into the intermediary's wealth formula $\bm{\delta_b} \times (\bm{P_b} - \bm{S})$ yields
	\begin{equation}\label{eq:intermed_prof}
		\mathcal{W}^{\mathrm{mar}}_{\mathrm{F}, b}  = \frac{(\mathrm{E}_b^{\mathrm{mar}})^2}{\varrho_b}.
	\end{equation}
	Substituting the value of $\tau -\bm{P_b}$ of~\eqref{eq:diff} into the wealth function ~\eqref{eq:taxmar} with $i=b$ yields
	\begin{align*}
		\mathcal{W}^{\mathrm{mar}}_{\mathrm{C}, b} & =  \displaystyle \mathcal{W}^{\mathrm{tax}}_{\mathrm{C}, b} +   \frac{ \left(\mathrm{E}_b^{\mathrm{tax}} + \mathrm{E}_b^{\mathrm{mar}} \right)}{2\varrho} \sum \mathrm{E}_\beta^{\mathrm{mar}}  -\frac{1}{2} \mathcal{W}^{\mathrm{mar}}_{{\mathrm{F}}, b} - \frac{1}{2 \varrho_b } \mathrm{E}_b^{\mathrm{tax}} \mathrm{E}_b^{\mathrm{mar}}            \\
		                                        & = \displaystyle \mathcal{W}^{\mathrm{tax}}_{\mathrm{C}, b} +   \frac{ \left(\mathrm{E}_b^{\mathrm{tax}} + \mathrm{E}_b^{\mathrm{mar}} \right)}{2 \varrho} \sum \mathrm{E}_\beta^{\mathrm{mar}}  -\frac{3}{2} \mathcal{W}^{\mathrm{mar}}_{{\mathrm{F}}, b} + \frac{\mathrm{E}_b^{\mathrm{mar}} }{2\varrho}  \sum \mathrm{E}_\beta^{\mathrm{mar}},
	\end{align*}
	where the second equality follows from substituting the identity (derive the expression of $\tau$ using~\eqref{eq:diff}):
	\begin{align*}
		\mathrm{E}_b^{\mathrm{tax}} & = \mathrm{E}_b^{\mathrm{bau}} - \tau \varrho_b = 2 \mathrm{E}_b^{\mathrm{mar}}  - \frac{\varrho_b}{\varrho }\sum \mathrm{E}_\beta^{\mathrm{mar}} .
	\end{align*}

	Aggregating~\eqref{eq:intermed_prof} over all intermediaries as $\mathcal{W}^{\mathrm{mar}}_{\mathrm{F}}   = \sum \mathcal{W}^{\mathrm{mar}}_{\mathrm{F}, b}$ yields the required expression for the aggregate wealth earned by financial intermediaries.
	As for the aggregate wealth of the companies, using the condition $\mathrm{E}^{\mathrm{tax}} = \sum\mathrm{E}_c^{\mathrm{tax}} + \sum \mathrm{E}_b^{\mathrm{tax}} =  \sum\mathrm{E}_c^{\mathrm{mar}} + \sum \mathrm{E}_b^{\mathrm{mar}} = \mathrm{E}^{\mathrm{mar}} =A$, and the aggregation
	\begin{align*}
		\mathcal{W}^{\mathrm{mar}}_\mathrm{C} & = \sum \mathcal{W}^{\mathrm{mar}}_{\mathrm{C}, c} + \sum \mathcal{W}^{\mathrm{mar}}_{\mathrm{C}, b}
	\end{align*}
	yields the required result for $\mathcal{W}^{\mathrm{mar}}_\mathrm{C}$.
	As for the wealth of the regulator under the market scheme, using the value of $\bm{S}$ in~\eqref{eq:S_tau} yields
	\begin{align*}
		\mathcal{W}^{\mathrm{mar}}_{\mathrm{R}} & = \bm{S} A=  \left( \tau -\frac{1}{\varrho}\sum \mathrm{E}^{\mathrm{mar}}_b \right) A
		= \mathcal{W}^{\mathrm{tax}}_{\mathrm{R}} - \frac{A}{\varrho} \sum  \mathrm{E}_b^{\mathrm{mar}},
	\end{align*}
	which ends the proof of the theorem.
\end{proof}

\begin{corollary}\label{cor:jensen_hybrid}
	Under the market scheme, the GDP of the real economy satisfies
	\begin{align*}
		\mathrm{GDP}^{\mathrm{mar}}
		 & = \mathrm{GDP}^{\mathrm{tax}} + \tfrac{1}{2}\underbrace{\left(  \tfrac{ \left(\sum \mathrm{E}^{\mathrm{mar}}_b\right)^2}{\varrho}-\sum \tfrac{(\mathrm{E}^{\mathrm{mar}}_b)^2}{\varrho_b} \right)}_{\leq 0}.
	\end{align*}
	Moreover, the combined GDP contributions of companies and the regulator are lower under the market scheme than under the tax scheme, i.e.
	\begin{align*}
			 	& \mathcal{W}^{\mathrm{mar}}_{\mathrm{C}} + \mathcal{W}^{\mathrm{mar}}_{\mathrm{R}}
			        		=  \mathcal{W}^{\mathrm{tax}}_{\mathrm{C}} + \mathcal{W}^{\mathrm{tax}}_{\mathrm{R}}+   \tfrac{3}{2} \left(  \tfrac{ \left(\sum \mathrm{E}^{\mathrm{mar}}_b\right)^2}{\varrho}-\sum \tfrac{(\mathrm{E}^{\mathrm{mar}}_b)^2}{\varrho_b} \right) - \tfrac{ \left(\sum \mathrm{E}^{\mathrm{mar}}_b\right)^2}{\varrho}.
	\end{align*}
\end{corollary}

\begin{proof}
	The inequality ($\leq 0$) holds by the application of Jensen-inequality.
\end{proof}

\section{Some Supplementary Results and Proofs on Section~\ref{sec:princing_random}}\label{app:random_case}

When company $i$ purchases certificates through an intermediary, the intermediated certificate price $\bm{P_i} = \bm{P_i}(S)$ given a spot price $S$, is characterized as follows:

\begin{proposition}\label{prop:spot_random}
	Let $S \in (0, \lambda]$ be a given spot price.
	If the maps
	\begin{equation*}
		[S, \lambda] \ni P \mapsto \mathbb{V}{\rm a}\mathbb{R}_{\frac{P}{\lambda}}(X)
		\quad \text{and} \quad
		P \mapsto P\, \mathbb{ES}_{\frac{P}{\lambda}}(X)
	\end{equation*}
	are continuously differentiable, then a unique optimal intermediated  price  $S \leq \bm{P_i}(S) \leq  \lambda$ exists and satisfies  the first-order condition
	\begin{equation*}
		\bm{\delta}'_i(\bm{P_i}) (\bm{P_i} - S) + \bm{ \delta_i}(\bm{P_i}) = 0,
	\end{equation*}
	for each company $i$, provided that $S<\lambda$.
\end{proposition}

\begin{proof}
	For readability, we omit the index $i$ in this proof and, for example, write $\bm{P}(S)$ instead of $\bm{P_i}(S)$.
	Since $\hat\varrho < \varrho$, from Assumption~\ref{ass:standing_ass}, it follows that
	\begin{equation}\label{eq:likeass}
		\lambda \hat{\varrho} < \mathrm{E}^{\mathrm{bau}}.
	\end{equation}
	The value-at-risk $\mathbb{V}{\rm a}\mathbb{R}_{\epsilon}(X)$  of a positive random variable $X$ is always positive for $0<\epsilon<1$.
	Hence, for $0<P<\lambda$, it holds
	\begin{equation*}
		0 < P\, \mathbb{ES}_{P/\lambda}(X) = \lambda \int_0^{P/\lambda} \mathbb{V}{\rm a}\mathbb{R}_{u}(X)du \leq \lambda \int_0^{1} \mathbb{V}{\rm a}\mathbb{R}_{u}(X)du =\lambda \mathbb{E}[X] = \lambda.
	\end{equation*}
	The identity $\int_0^1 \mathbb{V}{\rm a}\mathbb{R}_{u}(X)du = \mathbb{E}[X]$ follows from the fact that  $\mathbb{V}{\rm a}\mathbb{R}_{U}(X)$ and $X$ are equal in distribution for a standard uniform random variable $U$, \citep[see][pp. 222]{neil2015}. 
	By the inequality~\eqref{eq:likeass} and Proposition~\ref{prop:optimalquant_random}, this yields
	\begin{equation*}
		0< \mathbb{V}{\rm a}\mathbb{R}_{P/\lambda}(X) \left[ \mathrm{E}^{\mathrm{bau}} -  P\, \mathbb{ES}_{P/\lambda}(X) \hat{\varrho} \right] = \bm{\delta}(P),
	\end{equation*}
	which implies each company's certificates demand function $\bm{\delta}(P)$ is always positive.
	Hence, for each intermediary, the wealth function $h(P) = \bm{\delta}(P)(P - S)$ can be optimized over $P \in [S,\lambda]$, since $h(P) \leq 0$ for $P < S$ and  $S \leq \lambda$ by assumption.
	If $S = \lambda$, then the interval $[S,\lambda]$ reduces to the singleton $\{\lambda\}$.
	Consequently, the optimal solution is given by $P = \lambda$, which is unique.
	Computations show that $\bm{\delta}(\lambda)$ corresponds to the left endpoint of the support of the distribution of $X$, which is equal to $0$.
	It follows that $h(S) = h(\lambda) = 0$, $P \mapsto h(P)$ is positive on $(S, \lambda)$ and continuously differentiable on $[S, \lambda]$.
	Hence, by the extreme value theorem, the function $h$ attains its maximum at some point in $\bm{P}\in (S, \lambda)$ satisfying the first order condition
	\begin{equation*}
		\bm{\delta}'(\bm{P}) (\bm{P} - S) + \bm{\delta}(\bm{P}) = 0.
	\end{equation*}
	By monotonicity of $\bm{\delta}$ on $(S, \lambda)$, this equation must have a unique solution and hence the intermediated price is unique.
\end{proof}

\begin{proposition}\label{prop:unique_spot_random}
	Under the assumptions of Proposition~\ref{prop:spot_random} and with feasibility condition\footnote{
		That is, $\sum \bm{\delta_i}(\lambda) \leq  A < \sum \bm{\delta_i}(0+)$ .
	} on the number of certificates  $A$, under the market scheme, there exists a unique equilibrium spot price $\bm{S} \in (0, \lambda]$, and the corresponding effective\footnote{
		The effective carbon price faced by company $i$ is $\bm{P_i}(\bm{S}) = \bm{S}$ if it participates directly in the regulator's auction; otherwise, it is the intermediated price set by the financial intermediary.
	}
	carbon price $\bm{P_i}(\bm{S})$ faced by company $i$ is also uniquely determined.
\end{proposition}

\begin{proof}
	Let $S_1, S_2 \in (0, \lambda]$ with $S_1 \leq S_2$. Then, the relations $[S_1, \lambda] \supseteq [S_2, \lambda]$ and
	\begin{equation*}
		\bm{P_i}(S_1) = \argmax_{P \in [S_1, \lambda]} \bm{ \delta_i}( P) (P -S_1) \quad \text{and} \quad  \bm{P_i}(S_2) = \argmax_{P \in [S_2, \lambda]} \bm{ \delta_i}( P) (P -S_2)
	\end{equation*}
	implies that $\bm{P_i}(S_1) \leq \bm{P_i}(S_2)$, and therefore the intermediated price $\bm{P_i}(S)$ is non-decreasing in $S$ for the companies going through intermediary.
	Computation of derivatives of value-at-risk $\mathbb{V}{\rm a}\mathbb{R}_{\epsilon}(X)$ and expected shortfall $\mathbb{ES}_{\epsilon}(X)$  risk measures with respect to $\epsilon$ yield 
	\begin{equation*}
		\partial_{\epsilon} \mathbb{V}{\rm a}\mathbb{R}_{\epsilon}(X) = \frac{1}{G'\left( \mathbb{V}{\rm a}\mathbb{R}_{\epsilon}(X) \right)},  \qquad \partial_{\epsilon} \mathbb{ES}_{\epsilon}(X) = \frac{1}{\epsilon}\left[ \mathbb{V}{\rm a}\mathbb{R}_{\epsilon}(X) - \mathbb{ES}_{\epsilon}(X)\right], 
	\end{equation*}
	where $G$ is the survival function of the random variable $X$.  
	This implies that the derivatives of $\bm{ \delta_i}$ with respect to $P_i$ is given by 
	\begin{equation}\label{eq:derivative}
		\bm{\delta}'_i(P_i) = \frac{ \mathrm{E}^{\mathrm{mar}}_i}{\lambda G' \left(\mathbb{V}{\rm a}\mathbb{R}_{P_i/\lambda}(X) \right)}  - \hat{\varrho} \left( \mathbb{V}{\rm a}\mathbb{R}_{P_i/\lambda}(X)\right)^2  \leq 0.
	\end{equation}
	It follows that $\bm{ \delta_i}(P_i)$ is non-increasing in $P_i$.
	Consequently,  $\sum \bm{ \bm{\delta_i}}(\bm{P_i}(S))$ is non-increasing in $S$.
	Since $S \in (0, \lambda]$ and $A$ is chosen in the feasible demand interval, the aggregated demand function $\sum \bm{ \bm{\delta_i}}(\bm{P_i}(S))$ equates $A$ at a unique point $\bm{S} \in (0, \lambda]$, and hence $\bm{S}$ is a unique equilibrium spot price.
	By Proposition~\ref{prop:spot_random}, the corresponding effective carbon price $\bm{P_i}(\bm{S})$ for each company $i$ is  also uniquely determined.
\end{proof}

\subsection{Proof of Theorem~\ref{thm:no_intermediaries_random}}\label{app:no_intermediaries_random}
	Under the tax scheme, by Proposition~\ref{prop:optimalquant_random}, the expected aggregate emission is
	\begin{equation*}
		\mathrm{E}^{\mathrm{tax}} = \mathrm{E}^{\mathrm{bau}} - \hat{ \varrho} \tau.
	\end{equation*}
	Setting $\mathrm{E}^{\mathrm{tax}}= \hat{A}$ and solving for $\tau$ yields the desired expression of $\tau$.
	As for the companies' wealth, it is given by Proposition~\ref{prop:optimalquant_random}.
	The regulator wealth  also follows from its definition.
 
	Under the spot market scheme, all companies face a uniform price $\bm{P_i} =\bm{S}$.
	By Proposition~\ref{prop:optimalquant_random} we have
	\begin{equation*}
		\mathrm{E}^{\mathrm{mar}}_i  =\dfrac{\bm{ \delta_i}(\bm{S})}{\mathbb{V}{\rm a}\mathbb{R}_{ \bm{S}/\lambda }(X)},  \quad \text{and hence}  \quad
		\mathrm{E}^{\mathrm{mar}}    = \dfrac{\sum \bm{ \delta_i}(\bm{S})}{\mathbb{V}{\rm a}\mathbb{R}_{ \bm{S}/\lambda }(X)}.
	\end{equation*}
	The emission constraint $\mathrm{E}^{\mathrm{mar}} = \hat{A}$ and the market clearing condition imply
	\begin{equation}\label{eq:lam_wR}
		\hat{A} = A/ \mathbb{V}{\rm a}\mathbb{R}_{ \bm{S}/\lambda }(X) \quad \text{and hence} \quad \bm{S}/\lambda = G\left( A/\hat{A}  \right).
	\end{equation}
	Proposition~\ref{prop:optimalquant_random} for $\bm{P_i} =\bm{S}$  implies
	\begin{equation}\label{eq:emis_rand}
		\mathrm{E}^{\mathrm{mar}}_i     = \mathrm{E}^{\mathrm{bau}}_i - \bm{S}\, \mathbb{ES}_{ \bm{S}/\lambda}(X)\hat{\varrho}_i
		\quad
		\text{and} 
		\quad 
		\mathrm{E}^{\mathrm{mar}}  = \mathrm{E}^{\mathrm{bau}} - \bm{S}\, \mathbb{ES}_{ \bm{S}/\lambda}(X)\hat{ \varrho}.
	\end{equation}
	Since  $\tau = (\mathrm{E}^{\mathrm{bau}}-\hat{A})/\hat{\varrho} $ under the tax scheme, the second identity in Eqn.~\eqref{eq:emis_rand}  and  the emission constraint $\mathrm{E}^{\mathrm{mar}} = \hat{A} $ yields $S\, \mathbb{ES}_{ \bm{S}/\lambda}(X) = \tau$, which with~\eqref{eq:lam_wR} implies
	\begin{equation*}
		\bm{S}= \displaystyle \frac{ \tau}{
		\mathbb{ES}_{ G(A /\hat{A})}(X)}= \lambda\ G\left(A/\hat{A} \right).
	\end{equation*}
	Proposition~\ref{prop:optimalquant_random} shows that a tax rate $\tau = \bm{S}\, \mathbb{E}_{\bm{S}/\lambda}(X)$ produces the same outcomes as the market scheme with certificate prices $\bm{P_i}(\bm{S}) = \bm{S}$.
	Hence, the two schemes are equivalent for the companies. 
	As for the regulator's wealth under the spot market scheme, by~\eqref{revenue:regu:mar}, we obtain
	\begin{align*}
		\mathcal{W}^{\mathrm{mar}}_{\mathrm{R}}
		 & = \displaystyle \left[ \mathbb{ES}_{ \bm{S}/\lambda}(X) -\mathbb{V}{\rm a}\mathbb{R}_{ \bm{S}/\lambda }(X) \right] \bm{S} \sum  \mathrm{E}^{\mathrm{mar}}_i  + \bm{S} A \\
		 & =\displaystyle \left[ \mathbb{ES}_{ \bm{S}/\lambda}(X) -\mathbb{V}{\rm a}\mathbb{R}_{ \bm{S}/\lambda }(X) \right] \bm{S}\ \hat{A}  + \bm{S} A                                \\
		 & = \left(\frac{\tau}{\bm{S}} - \dfrac{A}{\hat{A}}\right) \bm{S} \hat{A} + \bm{S} A=  \mathcal{W}^{\mathrm{tax}}_{\mathrm{R}},
	\end{align*}
	where the last equality holds due to the identity $\tau = \bm{S}\, \mathbb{ES}_{ \bm{S}/\lambda}(X)$ and \eqref{eq:lam_wR}.

\subsection{Proof of Proposition~\ref{prop:with_intermediaries_random}}\label{app:with_intermediaries_random}
	When companies are identical, the tax $\tau$ given by Theorem~\ref{thm:no_intermediaries_random} simplifies to
	\begin{equation*}
		\tau = \displaystyle \frac{\mathrm{E}^{\mathrm{bau}} -\hat{A}}{ N\hat{\varrho}_1} \quad \text{with} \quad \hat{\varrho}_1 = \frac{1}{\pi^1_1} + \frac{1}{\gamma_1\sigma^2}.
	\end{equation*}
	Under the purely intermediated  market scheme, by symmetry, the companies  face a uniform intermediated price $\bm{P_i} =\bm{P}$.
	By Proposition~\ref{prop:optimalquant_random} we have
	\begin{equation*}
		\mathrm{E}^{\mathrm{mar}}_i  =\dfrac{\bm{ \delta_i}(\bm{P})}{\mathbb{V}{\rm a}\mathbb{R}_{ \bm{P}/\lambda }(X)},  \quad \text{hence}  \quad
		\mathrm{E}^{\mathrm{mar}}    = \dfrac{\sum \bm{ \delta_i}(\bm{P})}{\mathbb{V}{\rm a}\mathbb{R}_{\bm{P}/\lambda }(X)}.
	\end{equation*}
	The emission constraint $\mathrm{E}^{\mathrm{mar}} = \hat{A}$ and the market clearing condition imply
	\begin{equation}\label{eq:lam_wR2}
		\hat{A} = A/ \mathbb{V}{\rm a}\mathbb{R}_{\bm{P}/\lambda }(X) \quad \text{and hence} \quad \bm{P}/\lambda = G\left( A/\hat{A}  \right).
	\end{equation}
	Proposition~\ref{prop:optimalquant_random} for $\bm{P_i} =\bm{P}$  implies
	\begin{equation}\label{eq:emis_rand_sec}
		\mathrm{E}^{\mathrm{mar}}_i     = \mathrm{E}^{\mathrm{bau}}_i - \bm{P}\, \mathbb{ES}_{ \bm{P}/\lambda}(X)\hat{\varrho}_i
		\quad \text{and} \quad
		\mathrm{E}^{\mathrm{mar}}  = \mathrm{E}^{\mathrm{bau}} - \bm{P}\, \mathbb{ES}_{ \bm{P}/\lambda}(X)\hat{ \varrho}.
	\end{equation}
	Since  $\tau = (\mathrm{E}^{\mathrm{bau}}-\hat{A})/\hat{\varrho}$ under the tax scheme, the second identity in ~\eqref{eq:emis_rand_sec}  together with the emission constraint $\mathrm{E}^{\mathrm{mar}} = \hat{A} $ yields $\bm{P}\, \mathbb{ES}_{ \bm{P}/\lambda}(X) = \tau$, which with~\eqref{eq:lam_wR2} implies
	\begin{equation*}
		\bm{P_i} = \bm{P}= \displaystyle \frac{ \tau}{
		\mathbb{ES}_{ G(A /\hat{A})}(X)}= \lambda\ G\left(A/\hat{A} \right).
	\end{equation*}
	By Proposition~\ref{prop:optimalquant_random}, under the market scheme, we have
	\begin{equation*}
		\bm{ \delta_i}(\bm{P})  = \mathbb{V}{\rm a}\mathbb{R}_{ \bm{P}/\lambda}(X) \left[  \mathrm{E}^{\mathrm{bau}}_i - \bm{P}\ \mathbb{ES}_{ \bm{P}/\lambda}(X) \hat{\varrho}_1  \right].
	\end{equation*}
	The emission constraint $\mathrm{E}^{\mathrm{mar}} =\hat{A}$ and by symmetry, we have $\mathrm{E}^{\mathrm{mar}}_i =\hat{A}/N$.
	Therefore, by ~\eqref{eq:derivative} and \eqref{eq:lam_wR2}, 
	\begin{align*}
		\bm{ \delta_i}'(\bm{P})
		 & =   \tfrac{ \mathrm{E}^{\mathrm{mar}}_i}{\lambda\ G'\left( \mathbb{V}{\rm a}\mathbb{R}_{ \bm{P}/\lambda}(X) \right) } - \left[ \mathbb{V}{\rm a}\mathbb{R}_{ \bm{P}/\lambda}(X) \right]^2 \hat{\varrho}_1     \\
		 & = \tfrac{ \hat{A}/ N}{\lambda\ G'\left( A/\hat{A}\right) } - \left( A/\hat{A} \right)^2 \hat{\varrho}_1.
	\end{align*}
	By  the market clearing condition together with  the fact that $\bm{ \delta_i}$ is identical for each company, we obtain $\bm{ \delta_i}(\bm{P})= A/N$.
	This together with the  first order condition in Proposition~\ref{prop:spot_random} yields
	\begin{equation*}
		\bm{S}  = \bm{P} + \tfrac{\bm{ \delta_i}(\bm{P})}{\bm{ \delta_i}'(\bm{P})} = \bm{P} - \tfrac{\lambda A}{\lambda N\hat{\varrho}_1 \left( A/\hat{A}\right)^2  -  \tfrac{\hat{A}}{G'\left(A/\hat{A}\right)}}.
	\end{equation*}

	Under the tax scheme, each company~$i$ faces a tax rate $\tau = \bm{P_i}\ \mathbb{ES}_{\bm{P_i}/\lambda }(X)$, where $\bm{P_i} =\bm{P}$.
	By Proposition~\ref{prop:optimalquant_random}, this implies that the wealth of the companies for each scheme is identical, i.e.  $\mathcal{W}^{\mathrm{mar}}_{\mathrm{C}, i}=\mathcal{W}^{\mathrm{tax}}_{\mathrm{C}, i} $.

	As for the wealth of intermediaries under the intermediated  market scheme, substituting the optimal values  $\bm{ \delta_i}(\bm{P}) = A/N$, $\bm{P_i} =\bm{P}$, and  $\bm{S}$ into company $i$'s intermediary wealth expression  $\bm{ \delta_i}(\bm{P})(\bm{P}-\bm{S})$ yields the desired expression for $\mathcal{W}^{\mathrm{mar}}_{\mathrm{F}, i}$.

	The expression for the wealth of the regulator $\mathcal{W}^{\mathrm{tax}}_{\mathrm{R}}$ under the tax scheme is obvious by the definition.
	Under the market scheme, substituting the optimal value $\bm{P_i} = \bm{P}$ and $\lambda$ into~\eqref{revenue:regu:mar} yields
	\begin{align*}
		\mathcal{W}^{\mathrm{mar}}_{\mathrm{R}}
		 & =  \displaystyle \left[ \mathbb{ES}_{\bm{P}/ \lambda }(X) - \mathbb{V}{\rm a}\mathbb{R}_{ \bm{P}/\lambda }(X) \right] \bm{P} \sum \mathrm{E}^{\mathrm{mar}}_i + \bm{S} A, \\
		 & =   \left[ \mathbb{ES}_{ \bm{P}/\lambda}(X) - \mathbb{V}{\rm a}\mathbb{R}_{ \bm{P}/\lambda}(X) \right] \bm{P} \, \mathrm{E}^{\mathrm{mar}} + \bm{A}                 \\
		 & = \left[ \mathbb{ES}_{G(A/\hat{A})}(X) - \dfrac{A}{\hat{A}}\right] \bm{P} \hat{A} + \bm{S} A
		= \tau \hat{A} - \underbrace{(\bm{P} - \bm{S})A}_{=N \mathcal{W}^{\mathrm{mar}}_{\mathrm{F},1}},
	\end{align*}
	where the last equality holds due to the identity $\tau = \bm{P}\, \mathbb{ES}_{ \bm{P}/\lambda}(X)$ and \eqref{eq:lam_wR2}.

\bibliographystyle{abbrvnat}
\bibliography{biblio.bib}
\end{document}